 \documentclass[twocolumn,twoside]{IEEEtran}

\usepackage{amsmath,amssymb,amsthm,wasysym,epsfig,color,subfigure,graphicx,epstopdf}
\usepackage{amsmath,amssymb,amsthm,wasysym,epsfig,color,subfigure,graphicx,epstopdf,url,bm,nicefrac}
\usepackage[table]{xcolor}
\usepackage{bbold}
\usepackage{caption}

\usepackage{algorithm}
\usepackage{algorithmic}

\allowdisplaybreaks

\usepackage{hyperref}
\hypersetup{colorlinks=false}

\usepackage{accents}
\makeatletter
\def\wid{\check{{\cc@style\underline{\mskip9.5mu}}}}
\def\Wideubar{\underaccent{{\cc@style\underline{\mskip6mu}}}}
\makeatother

\makeatletter
\def\wideubar{\underaccent{{\cc@style\underline{\mskip9.5mu}}}}
\def\Wideubar{\underaccent{{\cc@style\underline{\mskip6mu}}}}
\makeatother

\makeatletter
\def\widebar{\accentset{{\cc@style\underline{\mskip9.5mu}}}}
\def\Widebar{\accentset{{\cc@style\underline{\mskip6mu}}}}
\makeatother

\newtheorem{lemma}{Lemma}

\newtheorem{theorem}{Theorem}
\newtheorem{definition}{Definition}

\theoremstyle{remark}

\def\ccalT{{\ensuremath{\mathcal T}}}

\begin{document}
\title{Compressive Phase Retrieval via\\
	Reweighted Amplitude Flow
}

\author{
 Liang Zhang,~\IEEEmembership{Student Member,~IEEE,}	 Gang Wang,~\IEEEmembership{Student Member,~IEEE,}\\ 
 Georgios B. Giannakis,~\IEEEmembership{Fellow,~IEEE}, and  Jie Chen,~\IEEEmembership{Senior Member,~IEEE}
	\thanks{The work of L. Zhang, G. Wang, and G. B. Giannakis in this paper was partially supported  by NSF grants 1423316, 1442686, 1508993, and 1509040. The work of J. Chen was partially supported by the National Natural Science Foundation of China grants U1509215, 61621063, and the Program for Changjiang Scholars and Innovative Research Team in University (IRT1208).
L. Zhang, G. Wang, and G. B. Giannakis are with the Digital Technology Center and the Department of Electrical and Computer Engineering at the University of Minnesota, Minneapolis, MN 55455, USA.  G. Wang is also with the State Key Laboratory of Intelligent Control and Decision of Complex Systems, Beijing 100081, P. R. China. 
J. Chen is with the School of Automation and the State Key Laboratory of Intelligent Control and Decision of Complex Systems, Beijing Institute of Technology, Beijing 100081, P. R. China.
E-mails: \{zhan3523, gangwang, georgios\}@umn.edu; chenjie@bit.edu.cn.}
}

\maketitle

\allowdisplaybreaks

\begin{abstract}
The problem of reconstructing a sparse signal vector from magnitude-only measurements (a.k.a., compressive phase retrieval), emerges naturally in diverse applications, but it is NP-hard in general. Building on recent advances in nonconvex optimization, this paper puts forth a new algorithm that is termed \emph{compressive reweighted amplitude flow} and abbreviated as CRAF, for compressive phase retrieval. Specifically, CRAF operates in two stages. The first stage seeks a sparse initial guess via a new spectral procedure.
In the second stage, CRAF implements a few hard thresholding based iterations using reweighted gradients. 
When there are sufficient measurements, CRAF provably recovers the underlying signal vector exactly with high probability under suitable conditions.
Moreover, its sample complexity coincides with that of the state-of-the-art procedures. Finally, substantial simulated tests showcase remarkable performance of the new spectral initialization, as well as improved exact recovery relative to competing alternatives.

\end{abstract}


\section{Introduction}\label{sec:intro}
Phase retrieval (PR) refers to the task of reconstructing a signal vector from its phaseless measured linearly transformed entries. It emerges naturally in a wide range of engineering and physics applications such as X-ray crystallography, astronomy, and coherent diffraction imaging~\cite{spm2016eldar}, \cite{siam2015candes}. In these setups, the physical sensors can only record the density (the number of photons) of the light waves, but not their phase. This missing phase information renders general phase retrieval  ill-posed. In fact, it has been established that reconstructing a discrete, finite-duration signal vector from its Fourier transform magnitudes is generally NP-complete \cite{nphard1}. 
 To obtain useful solutions, additional assumptions have to be made, which include  (block) sparsity of underlying signal vectors~\cite{1982fienup, gespar, sparta, as2016clm}, non-negativity \cite{1982fienup,spm2016eldar}, and random Gaussian measurements~\cite{altmin, wf, twf, taf,raf,reshaped1,2017yxc}.

A number of phase retrieval  approaches have been developed so far, a sample of which are reviewed next. Alternating projection methods were advocated in \cite{gerchberg}, \cite{1978fienup}. By means of matrix-lifting and upon dropping the nonconvex rank constraint, convex semidefinite programs (SDP) were formulated \cite{phaselift,phasecut}. Minimizing the least-squares or least-absolute-value loss, several iterative solvers were pursued, namely those abbreviated as AltMinPhase \cite{tsp2015njs}, Wirtinger flow (WF) \cite{wf, twf, reshaped1, pwf, as2016clm}, amplitude flow \cite{taf, staf, raf, nips2016wg},  
and composite optimization \cite{duchi2017solving}. Convex phase retrieval approaches without matrix lifting can be found in \cite{phasemax,sparsephasemax}. We also recently developed a reweighted amplitude flow (RAF) algorithm which benchmarks the numerical performance of phase retrieval of signal vectors from Gaussian random measurements \cite{raf}.

The aforementioned phase retrieval approaches do not exploit possible structural information of the underlying signal vector, and they require for exact recovery that the number of measurements be on the order of the dimension of the vector \cite{wf,raf}. This number in large-scale high-resolution imaging applications is on the order of millions, rendering such algorithms inefficient. The signal vectors or their feature maps in many practical setups however, are naturally sparse
or admit an (approximately) sparse representation after certain known and
deterministic linear transformations have been applied~\cite{spm2016eldar}. This prior information can be critical in reducing the number of measurements required by general phase retrieval  approaches, and has prompted the development of various (block) sparse phase retrieval  solvers. To obtain sparse solutions, the  $\ell_1$-regularized PhaseLift was solved in~\cite{cprl}. Targeting nonconvex compressive phase retrieval formulations, a greedy algorithm was devised            \cite{gespar}, and the soft-thresholded Wirtinger flow (TWF) \cite{as2016clm} as well as the sparse truncated amplitude flow (SPARTA) \cite{sparta} was developed; see also~\cite{2017hedge} for the (block) compressive phase retrieval with alternating minimization (CoPRAM).

Building upon and going well beyond our precursors in \cite{raf,sparta}, this paper puts forth a new algorithm termed \emph{compressive reweighted amplitude flow} (CRAF) for (block)-sparse phase retrieval. Generalizing~\cite{sparta}, while further accounting for the structured sparsity pattern, the amplitude-based (block)-sparse phase retrieval problem is formulated, and it is solved in two stages, namely the initialization and the refinement stages. To enhance the initialization performance, a new sparse spectral initialization is developed, which judiciously assigns a \emph{negative or positive weight} to each sample. As such, the mean of the resultant initialization matrix features a larger gap between the first and the second eigenvalues, hence yielding improved performance as will be demonstrated in the numerical tests. 
 The second stage of CRAF successively refines the initialization by means of (model-based) hard thresholding iterations using \emph{reweighted} gradients. 
From the theoretical side, CRAF provably recovers the true signal vector at a linear rate under suitable conditions. 
Finally, numerical tests showcase the CRAF's improved recovery, and robustness to unknown sparsity relative to competing approaches.

The remainder of this paper is structured as follows. Section \ref{sec:prob} outlines the (block)-sparse phase retrieval  problem. Section~\ref{sec:alg} describes the algorithm, and establishes its convergence. 
Simulated tests are presented in~Section~\ref{sec:numerical}, and the proofs of the main theorems are given in Section~\ref{sec:proof}. Section~\ref{sec:con} concludes the paper.

Regarding notation, lower- (upper-) case boldface letters stand for column vectors (matrices). Sets are represented by calligraphic letters,  e.g., $\mathcal{S}$, with the exception of $\ccalT$ as superscript denoting matrix or vector transposition. The cardinality of set $\mathcal{S}$ is given by $|\mathcal{S}|$. Symbol $\|\cdot\|_2$ is reserved for the Euclidean norm, whereas $\|\cdot\|_0$ for the $\ell_0$ (pseudo)-norm counting the number of nonzero entries in a vector. Operator $\lceil\cdot\rceil$ returns the smallest integer greater than or equal to the given scalar. The Gauss error function  ${\rm{erf}}(x)$ is defined as ${\rm{erf}}(x):=(1/\sqrt{\pi}) \int_{-x}^x e^{-\tilde{x}^2}d\tilde{x}$. For a positive integer $m$, $[m]$ denotes the index set $\{1,\,2,\, \ldots,\, m\}$. Finally, the ordered eigenvalues of matrix $\bm{X} \in \mathbb{R}^{n\times n}$ are given as  $\lambda_1(\bm{X})\ge \lambda_2(\bm{X})\ge \cdots \ge\lambda_n(\bm{X})$.

\section{Compressive Phase Retrieval}\vspace{-.em}\label{sec:prob}
The compressive phase retrieval  aims at recovering a sparse signal vector from a few magnitude-only measurements~\cite{gespar, sparta, as2016clm}. Mathematically, it can be described as follows: Given a small set of phaseless linear measurements
\begin{equation}\label{eq:quad}
\psi_i=|\langle\bm{a}_i,\bm{x}\rangle|,\quad 1\le i \le m
\end{equation}
in which $\{\psi_i\}_{i=1}^m$ are the observed magnitudes, and $\{\bm{a}_i\in\mathbb{R}^n\}_{i=1}^m$ the known sampling vectors, the goal is to recover a $(kB)$-sparse solution $\bm{x}\in\mathbb{R}^n$, namely $\|\bm{x}\|_0\le kB$ with $kB$ being the known sparsity level. 
To accommodate also the block-sparse signal vectors, the following terminology is useful. Suppose without loss of generality that $\bm x$ is split into $N_B$ blocks $\{\bm x_b\}_{b=1}^{N_B}$, namely one can write $\bm{x}:= [\bm{x}_1^\ccalT\, \cdots\, \bm{x}_{N_B}^\ccalT]^\ccalT$ \cite{tsp2017zwrg}. For notational brevity, let $\mathcal{N}_B:=\{1,\,\ldots, \,N_B\}$ denote the index set of all blocks, and  $\mathcal{B}_b$ collect all indices of the entries of $\bm x$ corresponding to the $b$-th block. Therefore, $\mathcal{B}_b \subseteq [n]$ for all $b\in\mathcal{N}_b$, where $[n]:= \{1, \,\ldots, \,n\}$ consisting of all indices of $\bm{x}$. 
 \begin{definition} [$k$-block-sparse vectors \cite{baraniuk2010model}]
The $k$-block sparse vectors refer to vectors
$\bm{x}= [\bm{x}_1^\ccalT\, \cdots\,\bm{x}_{N_B}^\ccalT]^\ccalT$ such that $\bm x_b = \bm 0$ for all $b \notin \mathcal{S}_B$, where $\mathcal{S}_B$ is a subset of $\mathcal{N}_B$ with cardinality $|\mathcal{S}_B| = k$. 
 \end{definition}
 For simplicity, we consider that each block of the signal vector has equal length, that is, $|\mathcal{B}_b |=  B$ for all $b\in\mathcal{N}_b$ with $BN_B=n$. 
 It is clear that when $B=1$, the block-sparse phase retrieval boils down to the ordinary or unstructured sparse phase retrieval. Accordingly, we will henceforth focus on developing recovery algorithms for a block-sparse signal vector. 

Adopting the least-squares criterion, the task of recovering a $k$-block sparse vector from $m$ magnitude-only measurements can be cast as  \cite{taf}
\begin{equation}
\label{eq:cost}
\underset{\bm{z}\in \mathcal{M}_B^k}{\text{minimize}}~~\ell(\bm{z}):=\frac{1}{2m}\sum_{i=1}^m\left(\psi_i-|\bm{a}_i^\ccalT\bm{z}|\right)^2
\end{equation}  
where $\mathcal{M}_B^k$ denotes the set of all $k$-block-sparse vectors of dimension~$n$. Because of the nonconvex objective and the combinatorial  constraint, the problem in~\eqref{eq:cost} is in general NP-hard, hence computationally intractable.

For analytical concreteness, we focus on the real Gaussian model, which assumes $\bm{x}\in\mathbb{R}^n$, and independent and identically distributed (i.i.d.)  sensing vectors follow $\bm{a}_i\sim\mathcal{N}(\bm{0},\,\bm{I}_n)$ for all $1\le i \le m$. When there are enough measurements, it is reasonable to assume existence of a unique (up to a global sign) $k$-block-sparse solution  $\{\pm\bm{x}\}$ to the quadratic system in \eqref{eq:cost}. The critical goal of this paper is to put forth simple and scalable algorithms that can provably reconstruct $\bm{x}$ from as few magnitude-only measurements as possible.

\section{Compressive Reweighted Amplitude Flow}
\label{sec:alg}
This section presents the two stages, namely the initialization and the gradient refinement stages of CRAF. To begin, the distance from any estimate $\bm{z}\in \mathbb{R}^n$ to the solution set $\{\pm\bm{x}\}\subseteq\mathbb{R}^n$ is defined as
${\rm dist}(\bm{z},\bm{x}):=\min\{\|\bm{z}+\bm{x}\|_2,\,\|\bm{z}-\bm{x}\|_2\}$.

\subsection{Sparse Spectral Initialization}\label{sub:init}
A modified spectral initialization that utilizes the information from all available data samples is delineated first. Relative to existing phase retrieval initializations suggested in \cite{wf,twf, taf,raf}, enhanced numerical performance is achieved by assigning judicious weights to all sampling vectors. Subsequently, the generalization of the new initialization procedure to compressive phase retrieval settings is justified.

\subsubsection{Spectral initialization}
Finding a good initialization is key in enabling strong convergence of iterative nonconvex optimization algorithms. Consider first the general phase retrieval, namely without exploiting the sparse prior information. Similar to past approaches, the new initialization entails estimating the norm  $\|\bm{x}\|_2$ as well as the directional vector $\bm{d}:=\bm{x}/\|\bm{x}\|_2$. Regarding the former, it has been well documented that the term $\hat{r}: = \sqrt{(1/m)\sum_{i=1}^m \psi_m^2}$ is an unbiased and tightly concentrated estimate of the norm $r:=\|\bm{x}\|_2$ when there are enough measurements \cite{wf}.
The challenge remains to estimate the direction $\bm{d}$, namely seek a unit vector $\hat{\bm {d}}$ that is maximally correlated with $\bm{d}$.

Among different initialization strategies, the procedure proposed in \cite{taf} proves successful in achieving excellent numerical performance in estimating $\bm{d}$; see also \cite{duchi2017solving} for robustified alternatives. However, the truncation therein discards the useful information carried over in a non-negligible portion of samples. 
To exploit all the data samples, the new spectral initialization obtains the wanted approximation vector as  
\begin{equation}
\label{eq:maxeig}
\hat{\bm{d}}:=\arg	\underset{\|\bm{z}\|_2=1}{\text{max}}~~
\bm{z}^\ccalT\Big(\frac{\lambda^-}{|{\mathcal{I}}^{-}|}\sum_{i\in{\mathcal{I}}^-}\bm{a}_{i}\bm{a}_i^\ccalT
+\frac{\lambda^+}{|{\mathcal{I}}^{+}|}\sum_{i\in \mathcal{I}^+}\bm{a}_{i}\bm{a}_i^\ccalT
\Big)\bm{z}
\end{equation}  
where $\lambda^-<0$ and  $\lambda^+>0$ are preselected coefficients, and the index sets ${\mathcal{I}}^{-}:= \{i\in [m]: \psi_i^2 \le \hat{r}^2/2\}$, and  ${\mathcal{I}}^{+}:=\{i \in [m]: \psi_i^2 \ge \hat{r}^2/2\}$.
It is worth pointing out that the judiciously devised index sets satisfy $\mathcal{I} = {\mathcal{I}}^{-} \cup {\mathcal{I}}^{+}$, suggesting full use of the available data samples.  With $\hat{r}$ and $\hat{\bm{d}}$ at hand, the initial estimate  of $\bm{x}$ can be obtained conveniently as $\bm{z}^0:= \hat{r}\hat{\bm{d}}$.

Intuitively, the initialization strategy in \eqref{eq:maxeig} can be justified as follows. Leveraging the rotational invariance of $\bm{a} \sim \mathcal{N}(\bm{0}, \bm{I})$, we have for any thresholds $\tau_1,\, \tau_2 \in [0,\,1]$:
\begin{align}
&\mathbb{E}\!\left[\bm{a}\bm{a}^\ccalT | \langle \bm{a}, \bm{d}\rangle^2\le \tau_1\right] \nonumber\\
&\quad\quad\quad\quad = \bm{I}_n - \bm{d}\bm{d}^\ccalT + \mathbb{E}[\langle \bm{a}, \bm{d}\rangle^2 | \langle \bm{a}, \bm{d}\rangle^2\le \tau_1] \bm{d}\bm{d}^\ccalT \label{eq:I+} \\
&\mathbb{E}\!\left[\bm{a}\bm{a}^\ccalT | \langle \bm{a}, \bm{d}\rangle^2\ge \tau_2\right] \nonumber\\
&\quad\quad\quad\quad= \bm{I}_n - \bm{d}\bm{d}^\ccalT + \mathbb{E}[\langle \bm{a}, \bm{d}\rangle^2 | \langle \bm{a}, \bm{d}\rangle^2\ge \tau_2] \bm{d}\bm{d}^\ccalT \label{eq:I-}.
\end{align}
 It has been proved in~\cite[Lemma 3.2]{duchi2017solving} that  $$\mathbb{E}\!\left[ \langle \bm{a}, \bm{d}\rangle^2 | \langle \bm{a}, \bm{d}\rangle^2\le \tau_1\right]\le \tau_1/3.$$ Therefore, the smallest eigenvalue of $\mathbb{E}[\bm{a}\bm{a}^\ccalT | \langle \bm{a}, \bm{d}\rangle^2\le \tau_1]$ satisfies
$$\lambda_n\big(\mathbb{E}[\bm{a}\bm{a}^\ccalT | \langle \bm{a}, \bm{d}\rangle^2\le \tau_1] \big) \leq \tau_1/3$$ whereas all other eigenvalues are
$$\lambda_i\big(\mathbb{E}[\bm{a}\bm{a}^\ccalT | \langle \bm{a}, \bm{d}\rangle^2\le \tau_1] \big)=1,\quad 1\le i\le n-1.$$ 

Similarly, one can establish the following lower bound for the second term  $\mathbb{E}[ \langle \bm{a}, \bm{d}\rangle^2 | \langle \bm{a}, \bm{x}\rangle^2\ge \tau_2]$ in~\eqref{eq:I-}.
\begin{lemma} \label{le:boundI+}
	Consider any nonzero signal vector $\bm{d} \in \mathbb{R}^n$ with $\|\bm{d}\|_2 = 1$. If $\bm{a} \sim \mathcal{N}(\bm{0}, \bm{I})$, then for any $\tau \ge 0$, it holds that    
	\begin{equation} \label{eq:ieqI+}
	\mathbb{E}\!\left[ \langle \bm{a}, \bm{d}\rangle^2\left | \langle \bm{a}, \bm{d}\rangle^2\right.\ge \tau\right]\ge 	\frac{6 -\tau\, \rm{erf}(\sqrt{\tau})}{6 -3 \,\rm{erf}(\sqrt{\tau})}.
	\end{equation}
\end{lemma}	
\begin{proof}
	With $\tilde a := \langle \bm{a}, \bm{d}\rangle$, it holds that $\tilde a \sim \mathcal{N}(0,1)$. 
	Let $\tilde a'$ be a random variable with the same density as $\tilde a$, and $p(\tilde a)$ and  $p(\tilde a')$  denote the density of $\tilde a$ and $\tilde a'$, respectively. It follows that 
	\begin{align*}
	&\mathbb{E}\!\left[ \langle \bm{a}, \bm{d}\rangle^2 | \langle \bm{a}, \bm{d}\rangle^2\ge \tau\right]  =\mathbb{E}[\tilde{a}^2| \tilde{a}^2 \ge \tau]\\ 
	&= \mathbb{E}\left[\tilde{a}^2|~ |\tilde{a}| \ge \sqrt{\tau}\right] 
    =  \int_{\sqrt{\tau}}^{\infty} \frac{\tilde{a}^2p(\tilde{a})}{\int_{\sqrt{\tau}}^\infty p(\tilde{a}') d\tilde{a}'} d\tilde{a} 	\\
    &= \frac{1-\int_{-\sqrt{\tau}}^{\sqrt{\tau}}\tilde{a}^2p(\tilde{a}) d\tilde{a}}{1-\int_{-\sqrt{\tau}}^{\sqrt{\tau}} p(\tilde{a}') d\tilde{a}'} 
	=  \frac{2-\mathbb{E}[ \tilde{a}^2 |  \tilde{a}^2\le \tau]\,\rm{erf}{(\sqrt{\tau}})}{2-\rm{erf}{(\sqrt{\tau}})} \\
	&\ge 
	\frac{6 -\tau\, \rm{erf}(\sqrt{\tau})}{6 -3\, \rm{erf}(\sqrt{\tau})} 
	\end{align*}
	where the last inequality relies on~\cite[Lemma 3.2]{duchi2017solving}.
\end{proof}

To help understanding the  assertion of Lemma~\ref{le:boundI+}, taking $\tau = 0.5$ as an example,  we find $\mathbb{E}[ \langle \bm{a}, \bm{d}\rangle^2 | \langle \bm{a}, \bm{d}\rangle^2\ge 0.5] \ge 1.42$  by substituting the inequality $ \rm{erf}(\sqrt{0.5}) \ge 0.68$ into \eqref{eq:ieqI+}. Hence, it holds that $\lambda_1(\mathbb{E}[ \langle \bm{a}, \bm{d}\rangle^2 | \langle \bm{a}, \bm{d}\rangle^2\ge 0.5])\ge 1.42$, and  $\lambda_i\big(\mathbb{E}[\bm{a}\bm{a}^\ccalT | \langle \bm{a}, \bm{x}\rangle^2\ge 0.5] \big)=1,~1\le i \le n-1$. 
Subsequently, it can be deduced that for $\lambda^-<0$ and $\lambda^+>0$, the largest eigenvalue of 
$\lambda^-\mathbb{E}[\bm{a}\bm{a}^\ccalT | \langle \bm{a}, \bm{d}\rangle^2\le 0.5] +\lambda^+ \mathbb{E}[\bm{a}\bm{a}^\ccalT | \langle \bm{a}, \bm{d}\rangle^2\ge 0.5]$ is greater than or equal to $1.42 \lambda^++0.166\lambda^- $, and all other eigenvalues equal  $\lambda^++\lambda^-$. Therefore, once the sample mean matrix $\frac{\lambda^-}{|{\mathcal{I}}^{-}|}\sum_{i\in{\mathcal{I}}^-}\bm{a}_{i}\bm{a}_i^\ccalT
+\frac{\lambda^+}{|{\mathcal{I}}^{+}|}\sum_{i\in \mathcal{I}^+}\bm{a}_{i}\bm{a}_i^\ccalT$ is sufficiently close to its mean, it would be possible to estimate $\bm{d}$ with high accuracy  based on the matrix perturbation lemma in \cite[Corollary 1]{yu2014useful}, which is also included as Lemma~\ref{le:mp} in the Appendix for completeness. The aforementioned arguments speak for the effectiveness of the proposed initialization, whereas the next theorem quantifies rigorously the initialization estimation error ${\rm dist}(\bm{z}^0, \bm{x})$.
\begin{theorem}\label{th:int}
Let $\bm{z}_0 = \hat{r}\hat{\bm{d}}$ with $\hat{\bm{d}}$ obtained from~\eqref{eq:maxeig}. For any given constant $\delta_0 \in (0, 1)$,
there exists numerical constants $c_0>0$ and $C_0$ such that 
 the following holds 
$${\rm dist}(\bm{z}_0,\bm{x}) \le \delta_0 \|\bm{x}\|_2 $$	
with probability at least  $1-10\exp(-c_0 m)$ when $m \ge C_0n $.
\end{theorem}  

For readability, the proof of Theorem~\ref{th:int} is deferred to Section~\ref{sec:thm1}. 
 Although the suggested initialization assumes a specific threshold $\hat{r}^2/2$ to split samples into $\mathcal{I}^-$ and $\mathcal{I}^+$, it is straightforward to incorporate two different thresholds $0\le \tilde{\tau}_1\leq \tilde{\tau}_2 \le 1$ such that  ${\mathcal{I}}^{-}:= \{i\in [m]: \psi_i^2 \le \tilde{\tau}_1\hat{r}^2\}$ and ${\mathcal{I}}^{+}:=\{i \in [m]: \psi_i^2 \ge \tilde{\tau}_2\hat{r}^2\}$. By appropriately selecting $\tilde{\tau}_1$ and $\tilde{\tau}_2$, the initialization performance can be further boosted. 
 It is worthing pointing out that the weak recovery performance of similar procedures has been studied in~\cite{mondelli2017fundamental}, which only provides guarantee for the case of $n\to\infty$.

\subsubsection{Support recovery}
The initialization procedure in~\eqref{eq:maxeig} is developed for general signal vectors $\bm{x}$, without leveraging the structural information that is present in diverse applications. When the vector is sparse, the required number of data samples to yield an accurate initialization can be reduced \cite{sparta}. Next, we demonstrate how to obtain a sparse initialization based on the procedure discussed in Section \ref{sub:init}. Similar to~\cite{sparta,2017hedge}, obtaining a sparse initialization entails first estimating the (block)-support of the underlying (block)-sparse signal vectors.

Specifically, define random variables $Z_{i,j}:=\psi_i^2a_{i,j}^2$, $\forall j\in[n]$.  According to \cite[Eq. (16)]{sparta}, the following holds
\begin{align} 
\mathbb{E}\Big[\sum_{j\in \mathcal{B}_b} Z_{i,j}^2\Big]
&=\mathbb{E}\Big[\sum_{j\in \mathcal{B}_b} (\bm{a}_i^\ccalT\bm{x})^4a_{i,j}^4\Big] \nonumber\\
&= 9B\|\bm{x}\|_2^4+24 \sum_{j\in \mathcal{B}_b} x_j^4+72\|\bm{x}_b\|^2\|\bm{x}\|_2^2\label{eq:blocksupp}.
\end{align}
If $b\in \mathcal{S}_B$, then $\bm{x}_b\ne \bm{0}$, yielding   
$\mathbb{E}\big[\sum_{j\in \mathcal{B}_b} Z_{i,j}^2\big]>9B\|\bm{x}\|_2^4+72\|\bm{x}_b\|^2\|\bm{x}\|_2^2$ in~\eqref{eq:blocksupp}. On the contrary, if $b\notin \mathcal{S}_B$, one has $\bm{x}_b= \bm{0}$, yielding $\mathbb{E}\big[\sum_{j\in \mathcal{B}_b} Z_{i,j}^2\big]= 9B\|\bm{x}\|_2^4$. It is evident that there is a separation of at least $72\|\bm{x}_b\|^2\|\bm{x}\|_2^2$ in the expected values of $\sum_{j\in \mathcal{B}_b} Z_{i,j}^2$ for $b\in \mathcal{S}_B$ and $b\notin \mathcal{S}_B$. As long as the gap $72\|\bm{x}_b\|^2\|\bm{x}\|_2^2$ is large enough, the (block)-support set $\mathcal{S}_B$ can be recovered exactly in this way. 

 
To estimate the (block)-support $\mathcal{S}_B$ in practice, compute first the so-called block marginals 
$$\zeta_{b}:=  {\sum_{j\in \mathcal{B}_b}\Big(\frac{1}{m}\sum_{i=1}^m\psi_i^2|a_{i,j}|^2\Big)^2}, \quad \forall b \in \mathcal{N}_b$$
which serves as an empirical estimate of $\mathbb{E}\big[\sum_{j\in \mathcal{B}_b} Z_{i,j}^2\big]$. As explained earlier, the larger $\zeta_{b}$ is, the more likely is for the block to be nonzero, namely $\|\bm{x}_b\|_2>0$~\cite{2017hedge}. 
 Upon collecting $\{\zeta_{b}\}_{b=1}^{N_B}$, one can pick the indices associated with the $k$-largest values in $\{\zeta_b\}_{b=1}^{N_B}$, which form the estimated block-support set $\hat{\mathcal{S}}_B$. Subsequently, an estimate of the support of $\bm{x}$ denoted as $\hat{\mathcal{S}}$ can be determined as 
 $\hat{\mathcal{S}}: = \{i\in \mathcal{B}_b \,|\, \forall b \in  \hat{\mathcal{S}}_B\}.$

 The support estimation procedure is summarized in Steps 2-4 of Algorithm~\ref{alg:SPARTA}.
Appealing to \cite[Theorem 5.1]{2017hedge} (also included as Lemma~\ref{le:supp} for completeness in the Appendix), Steps 2-4 recover the  support of $\bm{x}$ exactly with probability  at least $1-\frac{6}{m}$ provided that  $m\ge C_0' k^2B\log(mn)$ for some positive constant $C_0'$ and the minimum block 
$${x}_{\min}^B~:=~\min_{b\in \mathcal{S}_B} \|\bm{x}_b\|_2^2$$ is on the order of $(1/k) \|\bm x\|_2^2$, namely, $x_{\min}^B = (C_0''/k) \|\bm x\|_2^2$ for some number $C_0''>0$.

If the support has been exactly recovered, that is, $\hat{\mathcal{S}}= \mathcal{S}$, one can rewrite $\psi_i=|\bm{a}_i^\ccalT\bm{x}|=|\bm{a}_{i,\hat{\mathcal{S}}}^\ccalT\bm{x}_{\hat{\mathcal{S}}}|$ for all $i\in [m]$, where $\bm{a}_{i,\hat{\mathcal{S}}}\in\mathbb{R}^{kB}$ contains entries of $\bm{a}_i$ whose indices belong to $ \hat{\mathcal{S}}$; and likewise for $\bm{x}_{\hat{\mathcal{S}}}\in\mathbb{R}^{k}$.  Then, the proposed initialization in \eqref{eq:maxeig} can be applied to the dimensionality-reduced data $\{(\bm{a}_{i, \hat{\mathcal{S}}}, \psi_i) \}_{i=1}^m$ to obtain
\begin{align*}
&\hat{\bm{d}}_{\hat{\mathcal{S}}}:=\! \max_{\bm{z}\in\mathbb{R}^{kB}} \bm{z}^\ccalT\! \left(\frac{\lambda^-}{|{\mathcal{I}}^{-}|}\sum_{i\in{\mathcal{I}}^-}\!\bm{a}_{i,\hat{\mathcal{S}}}\bm{a}_{i, \hat{\mathcal{S}}}^\ccalT
+\frac{\lambda^+}{|{\mathcal{I}}^{+}|}\sum_{i\in \mathcal{I}^+}\!\bm{a}_{i, \hat{\mathcal{S}}}\bm{a}_{i,\hat{\mathcal{S}}}^\ccalT\right)\!\bm{z}.
\end{align*}
Subsequently, an estimate of the $n$-dimensional vector $\bm{d}$ can be constructed by zero-padding entries of $\hat{\bm{d}}_{\hat{\mathcal{S}}}$ whose indices do not belong to $\hat{\mathcal{S}}$. 

\begin{algorithm}[t]
  \caption{Compressive Reweighted Amplitude Flow (CRAF)}
  \label{alg:SPARTA}
  \begin{algorithmic}[1]
\STATE {\bfseries Input:}
Data $\{(\bm{a}_i;\psi_i)\}_{i=1}^m$, block length $B$, and block sparsity level $k$;
initialization parameters $\lambda^- = -3$ and $\lambda^+ = 1$; 
 step size $\mu=1$; and weighting parameters $\{\beta_i=0.6\}_{i=1}^m$, $\tau_w=0.1$ .
\STATE{\bfseries For} $b = 1$ {\bfseries to} $N_B$,  compute $$\zeta_{b}:= \sum_{j\in \mathcal{B}_b}\Big(\frac{1}{m}\sum_{i=1}^m\psi_i^2|a_{i,j}|^2\Big)^2.$$
\STATE{\bfseries Set} $\hat{\mathcal{S}}_B$ to include indices 
corresponding to the $k$-largest instances in $\{\zeta_b\}_{b=1}^{N_B}$.
\STATE{\bfseries Set} $\hat{\mathcal{S}}$ to comprise indices of $\mathcal{B}_b$ for $b \in \hat{\mathcal{S}}_B$.
\STATE {\bfseries Compute}\label{step:3}
 the principal eigenvector $\hat{\bm{d}}_{\hat{\mathcal{S}}}\in\mathbb{R}^{kB}$ of $$\frac{\lambda^-}{|{\mathcal{I}}^{-}|}\sum_{i\in{\mathcal{I}}^-}\bm{a}_{i,\hat{\mathcal{S}}}\bm{a}_{i, \hat{\mathcal{S}}}^\ccalT
 +\frac{\lambda^+}{|{\mathcal{I}}^{+}|}\sum_{i\in \mathcal{I}^+}\bm{a}_{i, \hat{\mathcal{S}}}\bm{a}_{i,\hat{\mathcal{S}}}^\ccalT$$
 where ${\mathcal{I}}^{-}:= \{i\in [m]: \psi_i^2 \le \hat{r}^2/2\}$  and  ${\mathcal{I}}^{+}:=\{i \in [m]: \psi_i^2 \ge \hat{r}^2/2\}$ with  $\hat{r}:=\sqrt{\sum_{i=1}^m\psi_i^2/m}$.
 \vspace*{.1cm}
\STATE {\bfseries Initialize}
\label{step:4} $\bm{z}^0$ as $\hat{r}\tilde{\bm{d}}$, where $\tilde{\bm{d}}\in\mathbb{R}^n$ is given by augmenting  $\hat{\bm{d}}_{\hat{\mathcal{S}}}$ in Step 5 with $\tilde{{d}}_i = 0$ for $i \notin \hat{\mathcal{S}}$.
  \STATE {\bfseries Loop: For}\label{step:5} 
  {$t=0$ {\bfseries to} $T-1$}
 \begin{equation*}\vspace{-.em}
	   	\bm{z}^{t+1}=\mathcal{H}_k^B\left(\bm{z}^t-\frac{\mu}{m}\sum_{i\in[m]}w_i^t\left(\bm{a}_i^\ccalT\bm{z}^t-\psi_i\frac{\bm{a}_i^\ccalT\bm{z}^t}{|\bm{a}_i^\ccalT\bm{z}^t|}\right)\bm{a}_i\right)
\end{equation*}
   where $w_i^t:= \max\left \{\tau_w,~\frac{|\bm{a}_i^\ccalT\bm{z}^t|}{|\bm{a}_i^\ccalT\bm{z}^t|+ \psi_i \beta_i}\right \}$. 
     \STATE {\bfseries Output:}
$\bm{z}^{T}$.
  \end{algorithmic}
\vspace{-.em}
\end{algorithm} 


\subsection{Refinement via Reweighted Gradient Iterations}
\label{sub:refine}
Upon obtaining an accurate initial point, successive refinements based on reweighted gradient iterations are effected. To account for the block-sparsity structure of the wanted signal vector $\bm{x}$, the model-based iterative hard thresholding (M-IHT)~\cite{baraniuk2010model} is invoked. To start, recall that the generalized gradient of the objective function in~\eqref{eq:cost} is \cite{taf}
\begin{equation}
\label{eq:tgg}
\nabla\ell(\bm{z}):=\frac{1}{m}\sum_{i\in[m]}\Big(\bm{a}_i^\ccalT\bm{z}-\psi_i\frac{\bm{a}_i^\ccalT\bm{z}}{|\bm{a}_i^\ccalT\bm{z}|}\Big)\bm{a}_i
\end{equation}
in which the convention ${\bm{a}_i^\ccalT\bm{z}}/{|\bm{a}_i^\ccalT\bm{z}|}:=0$ for $|\bm{a}_i^\ccalT\bm{z}| = 0$ is adopted.

With $t\ge 0$ denoting the iteration count and $\bm{z}^0$ being the initial point, the M-IHT algorithm proceeds with the following $k$-block-sparse hard thresholding, namely
 \begin{equation}\label{eq:iteration}
 \bm{z}^{t+1}=\mathcal{H}_k^B\!\left(\bm{z}^t-\frac{\mu}{m}\nabla\ell(\bm{z}^t)\right)
 \end{equation}
 where $\mu>0$ is the preselected step size, and the block-sparse hard thresholding operator  $\mathcal{H}_k^B(\bar{\bm{u}}):\mathbb{R}^n\to \mathbb{R}^n$ converts an $n$-dimensional vector $\bar{\bm{u}}:=[\bar{\bm{u}}_1^\ccalT \,\ldots \,\bar{\bm{u}}_{N_B}^\ccalT]^\ccalT$ into a $k$-block-sparse one $\bm{u}:=[\bm{u}_1^\ccalT \,\ldots \,\bm{u}_{N_B}^\ccalT]^\ccalT$ such that 
\[ \bm{u}_b =
\begin{cases}
\bar{\bm{u}}_b,      &  \text{if } b \in \mathcal{U}_B\\
\bm{0},  &  \text{if } b \notin \mathcal{U}_B
\end{cases}
\]
where $\mathcal{U}_B$ comprises indices corresponding to the $k$-largest entities in $\{\|\bar{\bm{u}}_b\|_2\}_{b=1}^{N_B}$. 

Unfortunately, the negative gradient $-\nabla\ell(\bm{z})$ may not drag the iterate sequence $\{\bm{z}^t \}$ to the global optimum $\bm{x}$ because the estimated sign $\bm{a}_i^\ccalT\bm{z}/|\bm{a}_i^\ccalT\bm{z}| $ in $\nabla\ell(\bm{z})$ may not coincide with the true one $ \bm{a}_i^\ccalT\bm{x}/|\bm{a}_i^\ccalT\bm{x}|$~\cite{taf}. As a consequence, the update in~\eqref{eq:iteration} may not always reduce the distance of the iterate to the global optimum. To alleviate the negative influence of the erroneously estimated signs, SPARTA implements the following truncated gradient $\nabla\ell_{\rm tr}(\bm{z}^t)$ \cite{sparta}
\begin{equation}
\label{eq:trun}
\nabla\ell_{\rm tr}(\bm{z}^t): =\frac{1}{m}\sum_{i\in\mathcal{I}^{t}}\Big(\bm{a}_i^\ccalT\bm{z}^t-\psi_i\frac{\bm{a}_i^\ccalT\bm{z}^t}{|\bm{a}_i^\ccalT\bm{z}^t|}\Big)\bm{a}_i
\end{equation}
where $$ \mathcal{I}^{t}:=\Big\{1\le i\le m\Big|\frac{|\bm{a}_i^\ccalT\bm{z}^t|}{|\bm{a}_i^\ccalT\bm{x}|}\ge \tau_{\rm g}\Big\}	$$ 
for some preselected truncation parameter. It is clear that $\nabla\ell_{\rm tr}(\bm{z})$ is based on data samples whose associated  $|\bm{a_i}^\ccalT\bm{z}|$ is of relatively large sizes. The reason for this gradient truncation is that gradients (summands in \eqref{eq:trun}) of large $|\bm{a_i}^\ccalT\bm{z}|/|\bm{a_i}^\ccalT\bm{x}|$ provably point toward the global optimum $\bm{x}$ with high probability~\cite{taf}. 
However, as pointed out in~\cite{raf}, the truncation operation may reject meaningful samples, which hampers the efficacy of  $\nabla\ell_{\rm tr}$ especially when the sample size is limited.

An alternative to the truncation trimming procedure is to introduce different weights for different gradients \cite{raf}, which helps fusing useful information from all gradient directions. Specifically, the ensuing reweighted gradient used in \cite{raf} proves successful in phase retrieval  of general signal vectors 
\begin{equation}
\nabla\ell_{\rm rw}(\bm{z}^t): =\frac{1}{m}\sum_{i\in[m]}w_i^t\Big(\bm{a}_i^\ccalT\bm{z}^t-\psi_i\frac{\bm{a}_i^\ccalT\bm{z}^t}{|\bm{a}_i^\ccalT\bm{z}^t|}\Big)\bm{a}_i 
\end{equation}
where the weights are given by
\begin{equation}
w_i^t:= \max\left \{\tau_w,~\frac{|\bm{a}_i^\ccalT\bm{z}^t|}{|\bm{a}_i^\ccalT\bm{z}^t|+ \psi_i \beta_i}\right \},\quad \forall i\in [m]
\end{equation}
 for certain preselected parameters $\tau_w>0$ and $\beta_i >0$ for all $i\in [m]$. Evidently, it holds that $\tau_w \leq w_i^t \leq 1$ for all $i\in[m]$, and the larger the ratio $|\bm{a_i}^\ccalT\bm{z}|/|\bm{a_i}^\ccalT\bm{x}|$, the larger the weight $w_i^t $. In this sense, $w_i^t$ reflects the confidence in the $i$-th negative gradient pointing toward the global optimum $\bm{x}$. 
 
 In the context of phase retrieval  of block-sparse vectors, it is thus reasonable to implement the M-IHT based iteration  using reweighted gradients, namely
 \begin{equation}\label{eq:biteration}
 \bm{z}^{t+1}:=\mathcal{H}_k^B\!\left(\bm{z}^t-\mu \nabla\ell_{\rm rw}(\bm{z}^t)\right).
 \end{equation}
The proposed block-sparse phase retrieval  solver is summarized in Algorithm~\ref{alg:SPARTA}, whose exact recovery is established in the following theorem. 
\begin{theorem}
\label{le:thresh}
 Let $\bm{x}\in \mathbb{R}^n$ be any $k$-block-sparse ($kB\ll n$) vector with $x_{\rm min}^B: = (C_0''/k) \|\bm x\|_2^2$.  Consider noiseless measurements $\{\psi_i = |\bm{a}_i^\ccalT\bm{x}|\}_{i=1}^m$ from the real Gaussian model.
 If $m\ge C_1 k^2B\log (mn)$, there exists a constant learning rate $\mu>0$, such that the successive estimates $\bm{z}^t$ in  Algorithm~\ref{alg:SPARTA} obey 
 \begin{equation}
 \|\bm{z}^t-\bm{x}\|_2\le \delta_0\rho^t\|\bm{x}\|_2
,\quad t=0,\,1,\,\ldots
 \end{equation} 
 with probability at least $1-c_2\exp(-c_1 m)-6/m$. Here, $0<\delta_0<1$, $0<\rho<1$, $\mu,\,\,c_1>0,\,c_2>0$, $C_0''$, and $C_1$ are certain numerical constants. 
\end{theorem}
The proof of Theorem \ref{le:thresh} is provided in Section~\ref{sec:thm2}. Regarding its implication, a couple of observations come in order. To start, as soon as $m\ge C_1 k^2B\log (mn)$,
CRAF recovers exactly $k$-block-sparse vectors $\bm{x}$ of non-negligible blocks. This sample complexity is consistent with the Block CoPRAM method in \cite{2017hedge}. Furthermore. CRAF converges exponentially fast. Expressed differently, it takes CRAF at most $T:=  \mathcal{O}(\log(1/\epsilon))$ iterations to reach a solution of $\epsilon$-relative accuracy.

 \section{Numerical Tests} 
\label{sec:numerical}
This section demonstrates the efficacy of the proposed initialization and the CRAF algorithm relative to the state-of-the-art approaches for sparse phase retrieval, including SPARTA~\cite{sparta} and CoPRAM~\cite{2017hedge}. 
In all experiments, the support $\mathcal{S}$ of the true signal vectors $\bm{x}\in \mathbb{R}^{3,000}$ was randomly chosen. The nonzero entries were generated using $\bm{x}_{\mathcal{S}}\sim \mathcal N{(\bm{0}, \bm{I})}$. The obtained  $\bm{x}$ was subsequently normalized such that $\|\bm{x}\|_2 =1$. The sampling vectors were generated using $\bm{a}_i\sim \mathcal N{(\bm{0}, \bm{I})}$, $1\le i\le m$.  For SPARTA, its suggested parameters were used.
 The parameters of CRAF were set as $\lambda^- = -3$, $\lambda^+ = 1$, $\{\beta_i = 0.6\}_{i=1}^m$, $\tau_w = 0.1$, and $\mu = 1$. For all simulated algorithms, the maximum iterations were fixed to  $T=1,000$, and all reported results are averaged over $100$ Monte Carlo simulations.

The first experiment evaluates the performance of our initialization relative to that in SPARTA~\cite{sparta} and CoPRAM~\cite{2017hedge} for block length $B=1$. Figure~\ref{fig:int} depicts the average relative error of the three initialization schemes
 with the sparsity level $k$ varying from $25$ to $35$, and $m/k$ fixed to  $30$. Clearly, the new initialization  outperforms the other two with large margins.

\begin{figure}[ht]
	\centering
	\includegraphics[scale=0.58]{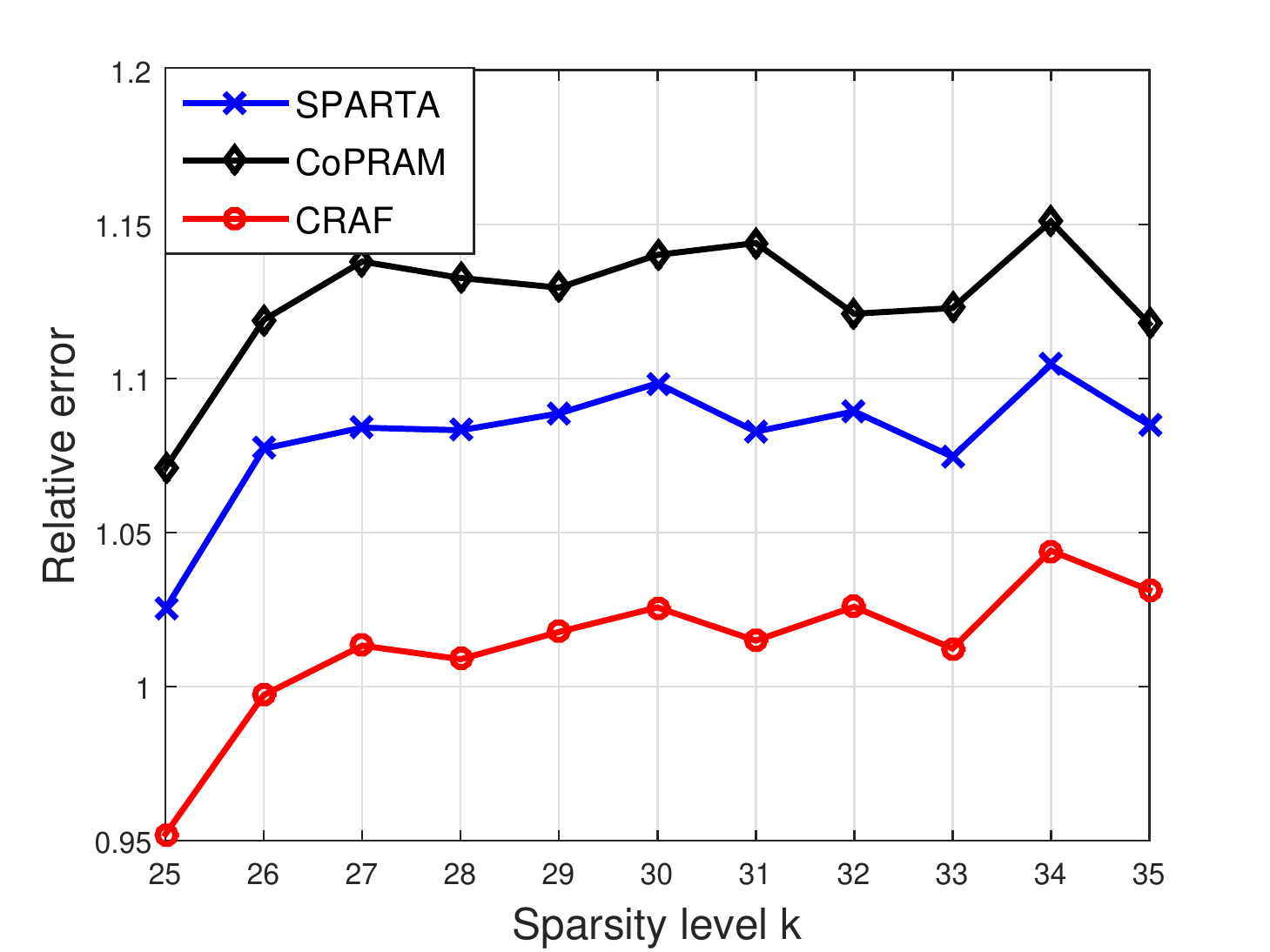}
	\caption{Average relative error for $B=1$.}
	\label{fig:int}
\end{figure}

The second experiment  examines the empirical success rates of CRAF, SPARTA, and CoPRAM  for solving the ordinary compressive phase retrieval with $B =1$.  Each of the $100$ Monte Carlo trials is declared a success if the relative error 
${\rm dist}(\bm{z}^T,\,\bm{x})/\|\bm{x}\|_2$ is less than $10^{-5}$. The empirical success rates of CRAF, SPARTA, and CoPRAM are presented in Fig.~\ref{fig:2} with $m$ increasing from $400$ to $1,800$. Notably, the curves showcase improved exact recovery performance of CRAF relative to its competing alternatives. Since in certain applications, the sparsity level $k$ may not be accurately known,  it is desirable to have the compressive phase retrieval  algorithms remain operational for unknown or inexact $k$ values. Let $\hat{k}$ be an estimate of the sparsity level $k$. 
The recovery performance of CRAF is tested with $\hat{k}$ set as the upper limit of the theoretically affordable sparsity level, namely $\sqrt{3,000}\approx 55$. 
From Fig.~\ref{fig:3}, it is clear that CRAF offers the best numerical performance for unknown $k$. A careful comparison between Figs.~\ref{fig:2} and \ref{fig:3} demonstrates that CRAF is more robust to unknown $k$ values than CoPRAM. 

\begin{figure}[ht]
	\centering
	\includegraphics[scale=0.58]{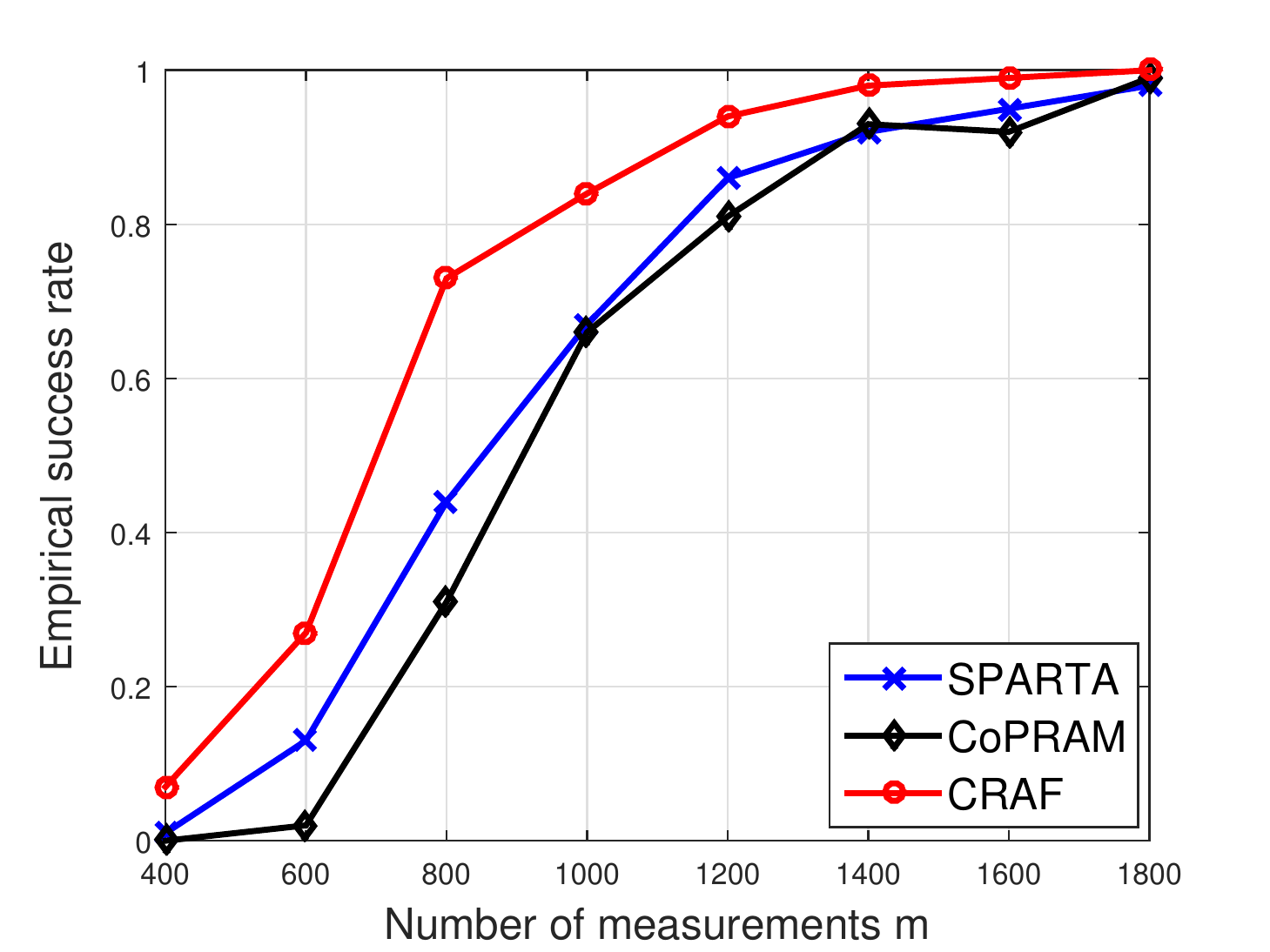}
	\caption{Empirical success rate versus $m$ for $B=1$, $k = 30$.}
	\label{fig:2}
\end{figure}

\begin{figure}[ht]
	\centering
	\includegraphics[scale=0.58]{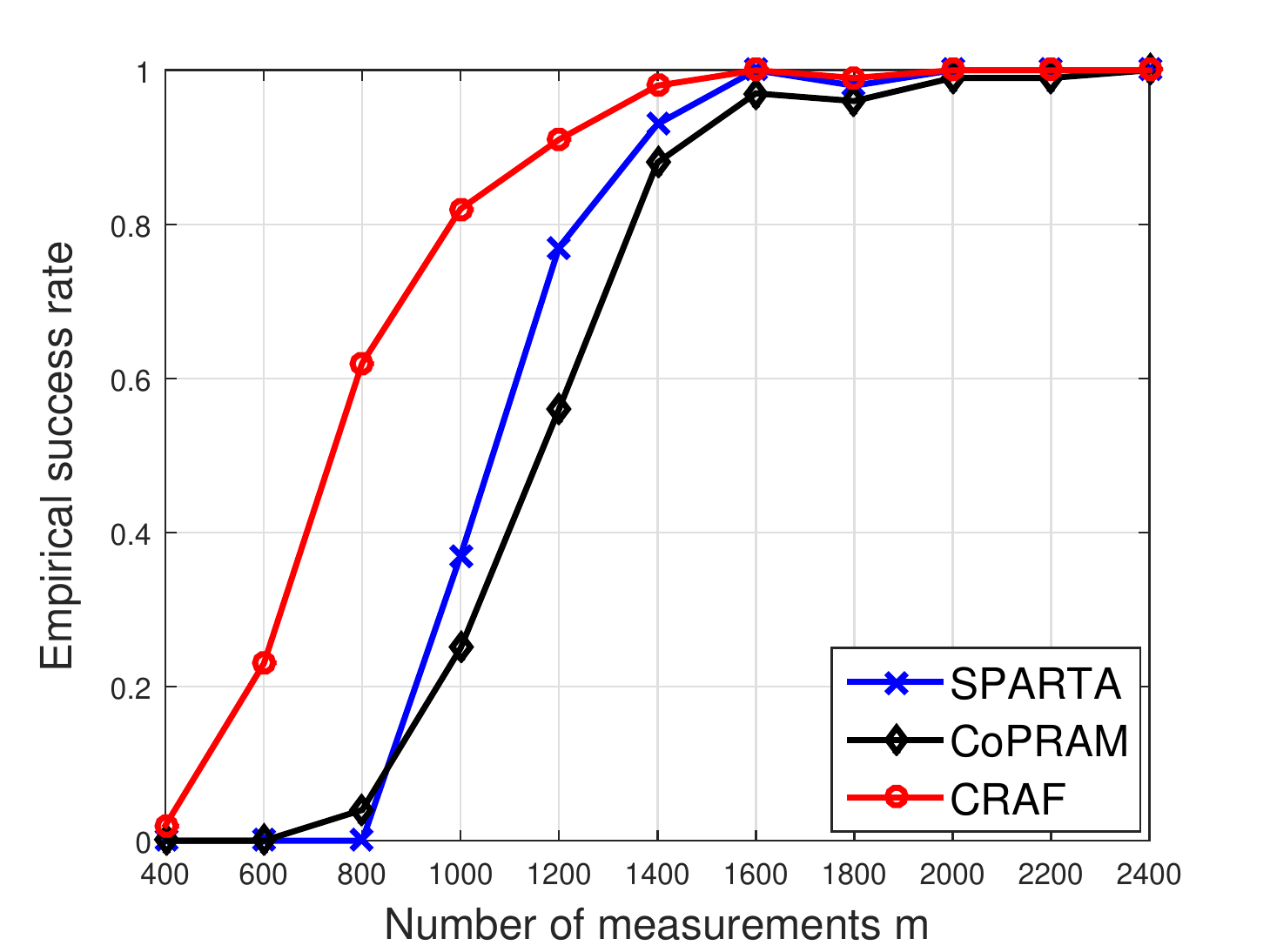}
	\caption{Empirical success rate versus $m$ for $B=1$, $k = 30$, and $\hat k = 55$.}
	\label{fig:3}
\end{figure}

The next experiment tests the performance of CRAF relative to that of Block CoPRAM and SPARTA for $20$-block-sparse phase retrieval with block length $B=2$. The empirical success rates for the three schemes from 100 independent trials with $k$ known and unknown are reported in Figs.~\ref{fig:4} and \ref{fig:5}, respectively. In both cases, CRAF yields the best recovery performance.

\begin{figure}[ht]
	\centering
	\includegraphics[scale=0.58]{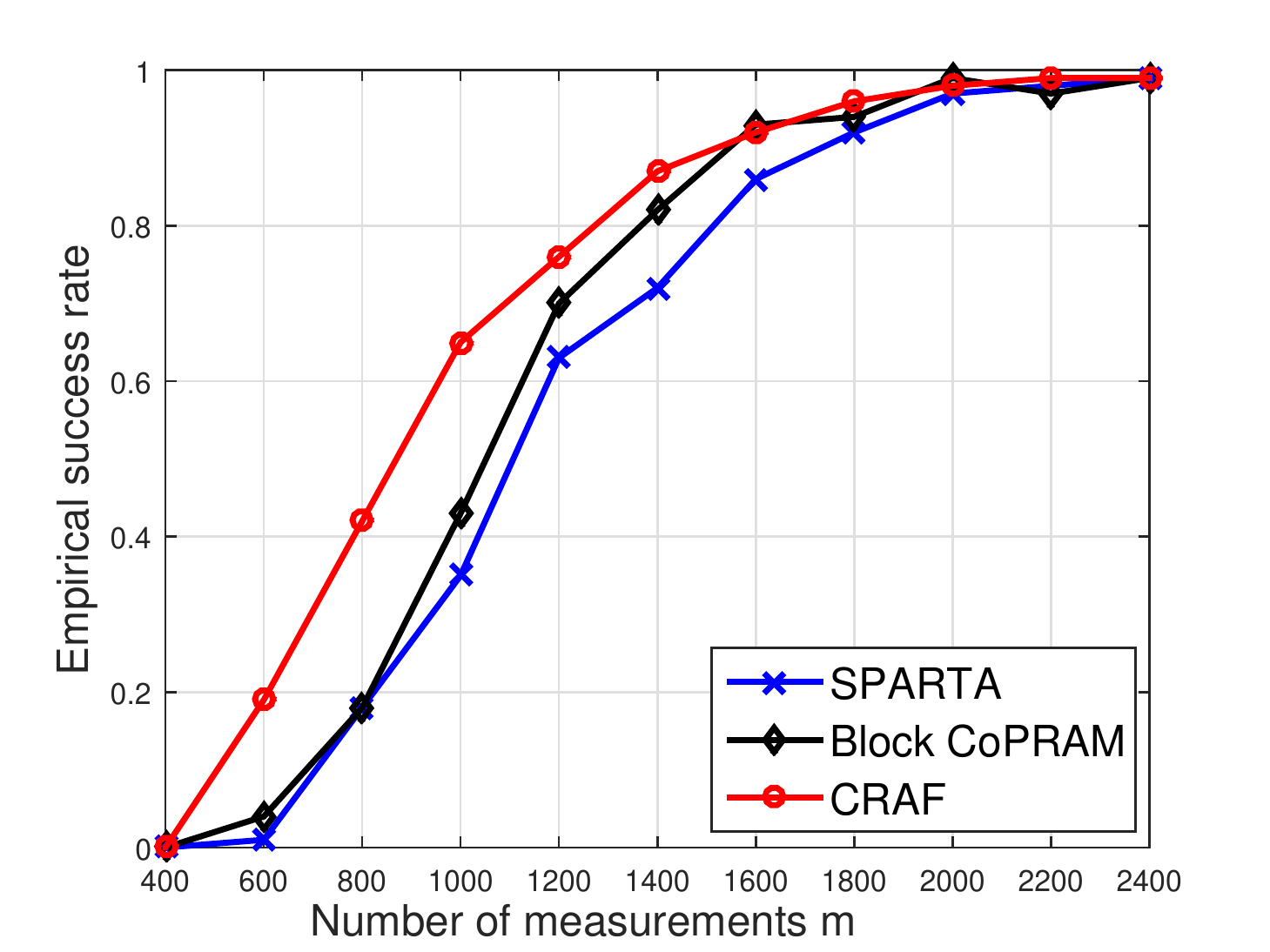}
	\caption{Empirical success rate versus $m$ for $B=2$, $k = 20$.}
	\label{fig:4}
\end{figure}

\begin{figure}[ht]
	\centering
	\includegraphics[scale=0.58]{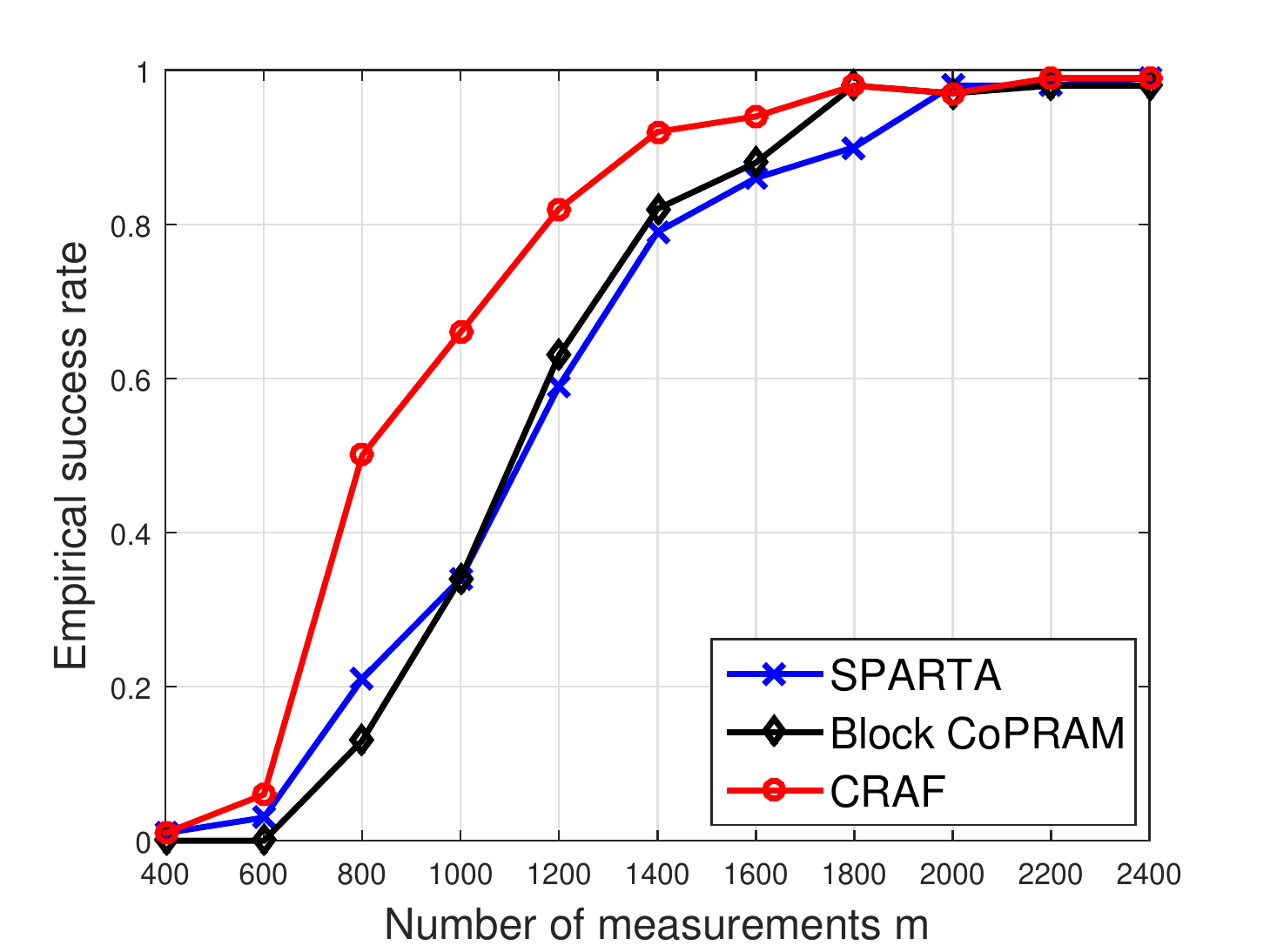}
	\caption{Empirical success rate versus $m$ for $B=2$, $k=20$, and $\hat{k}=28$.}
	\label{fig:5}
\end{figure}

The last experiment validates the robustness of CRAF with respect to noisy measurements of the following form:
$$\psi_i = |\bm{a}_i^\ccalT \bm{x}|+\eta_i,\quad 1\le i\le m $$  
where $\{\eta_i\}$ are independently sampled from $\mathcal{N}(0, \sigma^2)$. In this experiment, $k =30$, $B=1$, and $m =1,600$ were simulated.
 Figure~\ref{fig:6} depicts the relative errors of the three approaches versus varying $\sigma^2$ from $0.1$ to $0.6$, from which it is clear that CRAF offers the most accurate estimates for all noise levels. In other words, CRAF achieves improved robustness relative to SPARTA and CoPRAM.

\begin{figure}[ht]
	\centering
	\includegraphics[scale=0.58]{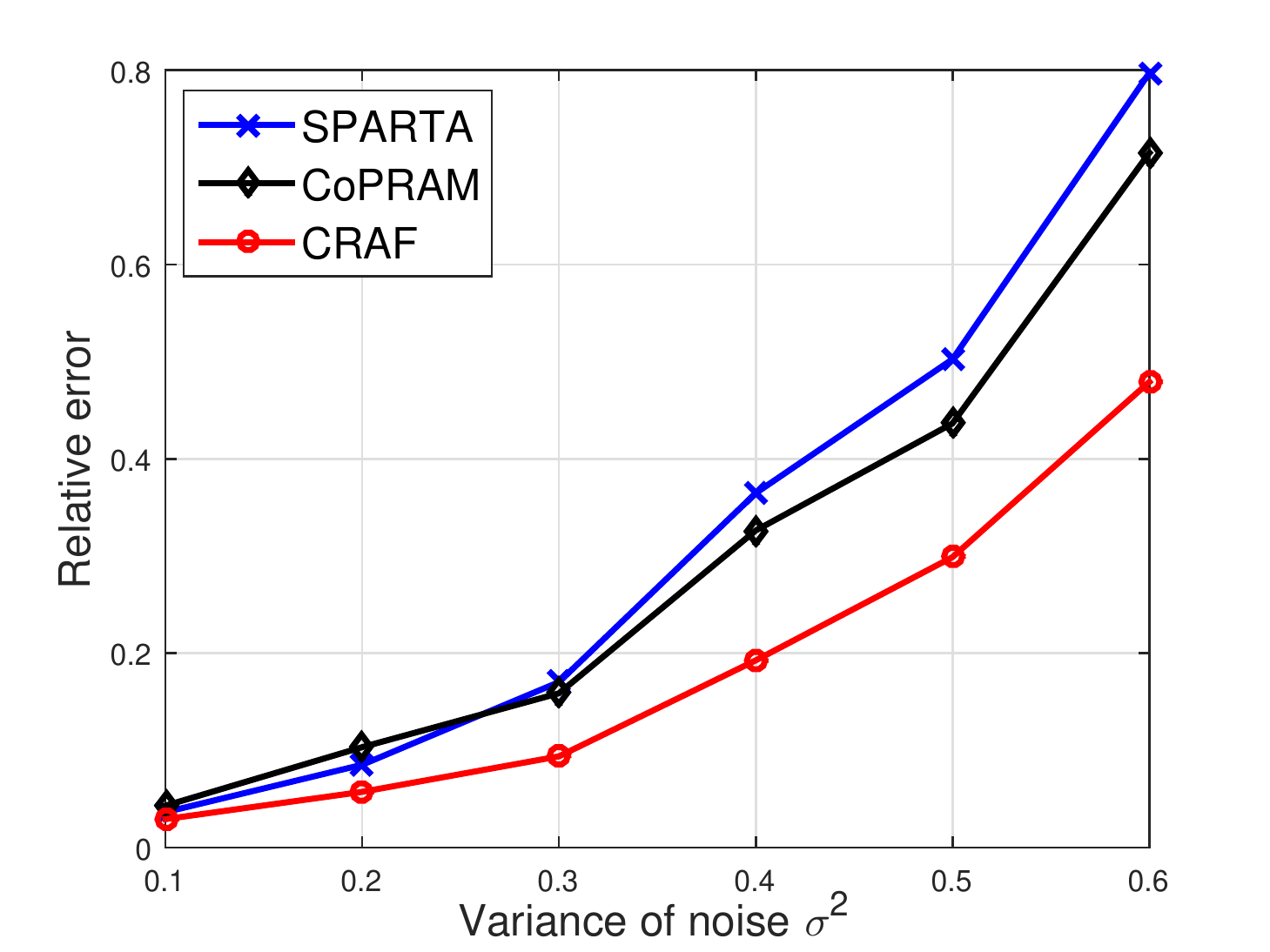}
	\caption{Average relative error of signal recovery versus the variance of noise $\sigma^2$  for $B=1$, $k=30$, and $m=1600$.}
	\label{fig:6}
\end{figure}

\section{Proofs}
\label{sec:proof}

The proofs of Theorem~\ref{th:int} and~\ref{le:thresh} are presented next.  The proof of Theorem~\ref{th:int} is mainly based on that of~\cite[Prop.~2]{duchi2017solving}, whereas the proof of Theorem~\ref{le:thresh} builds upon \cite[Thm. 1]{sparta}.
\subsection{Proof of Theorem 1}\label{sec:thm1}

	For ease of presentation, some notation is established first. To begin, 
	let $$\bm{M}^-:= \frac{1}{|{\mathcal{I}}^{-}|}\sum_{i\in{\mathcal{I}}^-}\bm{a}_{i}\bm{a}_i^\ccalT,~\text{and}~\bm{M}^+:= \frac{1}{|{\mathcal{I}}^{+}|}\sum_{i\in{\mathcal{I}}^+}\bm{a}_{i}\bm{d}_i^\ccalT$$ denote the first and second parts of the matrix used in the new initialization procedure \eqref{eq:maxeig}.
	Upon defining
	\begin{align*}
	\phi(\tau^-) &:= \mathbb{E}[\langle \bm{a}, \bm{d}\rangle^2 | \langle \bm{a}, \bm{d}\rangle^2\le \tau^-] -1\\
	\phi(\tau^+) &:=  \mathbb{E}[\langle \bm{a}, \bm{d}\rangle^2 | \langle \bm{a}, \bm{d}\rangle^2\ge \tau^+] -1 
	\end{align*}
    it can be verified that 
    $$\mathbb{E}[\bm{M}^-] = \bm{I}_n +\phi(\tau^-)\bm{d}\bm{d}^\ccalT\!,~\text{and}~\mathbb{E}[\bm{M}^+] = \bm{I}_n +\phi(\tau^+)\bm{d}\bm{d}^\ccalT.$$
	Without loss of generality, one can then write 
	\begin{align*}
	\bm{M}^-&=\bm{I}_n +\phi^-(\epsilon) \bm{d}\bm{d}^\ccalT+\bm{\Delta}^-\\  
	\bm{M}^+&=\bm{I}_n +\phi^+(\epsilon) \bm{d}\bm{d}^\ccalT+\bm{\Delta}^+
	\end{align*} 
where  $\bm{\Delta}^-$ $(\bm{\Delta}^+)$ describes the difference between $\bm{M}^-$ $(\bm{M}^+)$ and their mean $\mathbb{E}[\bm{M}^-]=\bm{I}_n +\phi^-(\epsilon) \bm{d}\bm{d}^\ccalT$ $(\mathbb{E}[\bm{M}^+]=\bm{I}_n +\phi^+(\epsilon) \bm{d}\bm{d}^\ccalT)$. Appealing to a standard eigenvector perturbation result \cite[Corollary 1]{yu2014useful} (also included as Lemma \ref{le:mp} in the Appendix), to bound ${\rm dist}(\hat{\bm{d}}, \bm{d})$, it suffices to bound   $\|\lambda^-\bm{\Delta}^-+\lambda^+\bm{\Delta}^+\|_2$.
	
	It has been established in~\cite[Proposition~2]{duchi2017solving} that for arbitrarily small $\delta^-\in (0, 1)$, and some absolute constants $c^->0$,  $C^-<\infty$ dependent on $\delta^-$, the following holds with probability at least $1-5\exp(-c^-m)$ 
	\begin{equation}
	\label{eq:delta2-}
	\|\bm{\Delta}^-\|_2 \leq \delta^-
	\end{equation}
	 whenever $m\ge C^- n$.
	 
	The task now remains
	to bound $\|\bm{\Delta}^+\|_2$. To account for the estimation error in $\hat r$, the following two index sets surrounding $\mathcal{I}^+$ are introduced \cite[Proposition 2]{duchi2017solving}  
	$$\mathcal{I}_{-\epsilon}^+ := \left\{ i \in [m]~|~\langle \bm{a}_i , \bm{d}\rangle^2 > \frac{1-\epsilon}{2} \right\}$$
	$$\mathcal{I}_{+\epsilon}^+ := \left\{ i \in [m]~|~\langle \bm{a}_i , \bm{d}\rangle^2 \ge \frac{1+\epsilon}{2} \right\}$$
	for some numerical constant $\epsilon \in (0, 1)$. 
	For convenience, one can set $\kappa :=  \exp((\epsilon -1)/4)/\sqrt{\pi(1-\epsilon)}$ which  upper bounds $P\!\left(\langle \bm{a}, \bm{d}\rangle^2 \in [\frac{1-\epsilon}{2}, \frac{1+\epsilon}{2}] \right)/\epsilon$, and $$p_0(\epsilon):= P\!\left(\langle \bm{a}, \bm{d}\rangle^2\ge \frac{1+\epsilon}{2}\right).$$ It can be readily checked that 
	\begin{align*}
	p_0(\epsilon)&= P\left(\chi_1^2\ge \frac{1+\epsilon}{2}\right) = 1- P\!\left(\chi_1^2\le \frac{1+\epsilon}{2}\right) \\
	&\ge 1-\sqrt{\frac{1+\epsilon}{2}\exp\!\left(\frac{1-\epsilon}{2}\right)}
	\end{align*}
	 leveraging the tail bound of the $\chi_1^2$ distribution.

	Subsequently, five events denoted as 
$\{\mathcal{E}_i\}_{i=1}^5$	
 occurring with high probability are introduced  in~\eqref{eq:events}.
	The constants $p\ge 1$ and $q\ge 1$ in~\eqref{eq:events}  satisfy $1/p+1/q=1$. On the event $\mathcal{E}_1$, it can be verified that $\hat{r}^2\in [(1-\epsilon)r^2, (1+\epsilon)r^2]$, hence   $\mathcal{I_{+\epsilon}^+}\subset \mathcal{I^+} \subset \mathcal{I_{-\epsilon}^+}$. When all five events $\{\mathcal{E}_i\}_{i=1}^5$  occur, $\|\bm{\Delta}^+\|_2$ can be bounded. In detail,  $\bm{\Delta}^+$ can be rewritten as in~\eqref{eq:delta+},
   		\begin{figure*}[!t]
   			\normalsize
            	\begin{align} 
            	&\mathcal{E}_1: = \left \{\frac{1}{m}\left\|\bm{A}^\ccalT\bm{A}\right \|_2  \in [1-\epsilon, 1+\epsilon] \right \},
            	\quad \qquad \mathcal{E}_2: = \left \{\left|\mathcal{I}_{-\epsilon}^+\right| \leq |\mathcal{I}_{+\epsilon}^+|+2\epsilon\kappa m\right \}\nonumber \\
            	&\mathcal{E}_3: = \left \{|\mathcal{I_{+\epsilon}^+}| \ge \frac{1}{2} m p_0({\epsilon})\right \}, \qquad \mathcal{E}_4: = \left \{\Big \|\frac{1}{m}\sum_{i \in \mathcal{I}_{-\epsilon}^+ \setminus  \mathcal{I}_{+\epsilon}^+}\bm{a}_i\bm{a}_i^\ccalT \Big\|_2 \le 4q(\kappa\epsilon)^{\frac{1}{p}}\right \}  \label{eq:events}\\
            	&\mathcal{E}_5: = 
            	 \left \{\left \|\frac{1}{|\mathcal{I_{+\epsilon}^+}|} \sum_{i \in \mathcal{I_{+\epsilon}^+}} \left[\bm{a}_i\bm{a}_i^\ccalT - \Big(\bm{I}_n +\phi^+\Big(\frac{1+\epsilon}{2}\Big)\bm{x}\bm{x}^\ccalT \Big)\right] \right\|_2  \leq \epsilon \right \} \nonumber
            	\end{align} 
            \hrulefill	
            
   			\begin{align}
   			\bm{\Delta}^+ &= \bm{M}^+- \Big(\bm{I}_n +\phi^+\Big(\frac{1+\epsilon}{2}\Big)\bm{d}\bm{d}^\ccalT \Big) \nonumber \\
   			& = \underbrace{\frac{1}{|\mathcal{I}_{+\epsilon}^+|}\sum_{i \in \mathcal{I_{+\epsilon}^+}} \left[\bm{a}_i\bm{a}_i^\ccalT - \Big(\bm{I}_n + \phi^+\Big(\frac{1+\epsilon}{2}\Big)\bm{d}\bm{d}^\ccalT \Big)\right] }_{:= \bm{\Delta}_1^+}
   			+\underbrace{\frac{1}{|\mathcal{I}_{+\epsilon}^+|} \sum_{i \in \mathcal{I}^+ \setminus  \mathcal{I}_{+\epsilon}^+}\bm{a}_i\bm{a}_i^\ccalT}_{:= \bm{\Delta}_2^+} 
   			~-~ \underbrace{\Big(\frac{1}{|\mathcal{I}_{+\epsilon}^+|}- \frac{1}{|\mathcal{I}^+|}\Big) \sum_{i\in \mathcal{I}^+}\bm{a}_i\bm{a}_i^\ccalT}_{:= \bm{\Delta}_3^+}  \label{eq:delta+}
   			\end{align}
   			\hrulefill
   		\end{figure*} 
	implying that $\bm{\Delta}^+ = \bm{\Delta}_1^+ + \bm{\Delta}_2^++ \bm{\Delta}_3^+$.  The three terms $\|\bm{\Delta}_1^+\|_2$, $\|\bm{\Delta}_2^+\|_2$, and $\|\bm{\Delta}_3^+\|_2$ are bounded next.
	Note that on the event $\mathcal{E}_5$, it holds that $\|\bm{\Delta}_1^+ \|_2\leq \epsilon$.  To bound $\|\bm{\Delta}_2^+ \|_2$, observe that on the event $\mathcal{E}_3$ and $\mathcal{E}_4$, the following are true
	\begin{align*}
	\|\bm{\Delta}_2^+\|_2& = \frac{m}{|\mathcal{I}_{+\epsilon}^+|}\Big \|   \frac{1}{m} \sum_{i \in \mathcal{I}^+ \setminus  \mathcal{I}_{+\epsilon}^+}\bm{a}_i\bm{a}_i^\ccalT\Big\|_2\\& \leq \frac{m}{|\mathcal{I}_{+\epsilon}^+|}\Big \|   \frac{1}{m} \sum_{i \in \mathcal{I}_{-\epsilon}^+ \setminus  \mathcal{I}_{+\epsilon}^+}\bm{a}_i\bm{a}_i^\ccalT\Big\|_2 
	\\&\le \frac{2}{p_0(\epsilon)}\cdot 4q(\kappa\epsilon)^{1/p}
	\end{align*}
	where the last inequality arises from the definitions of $\mathcal{E}_3$ and $\mathcal{E}_4$. Regarding $ \|\bm{\Delta}_3^+\|_2$, the next holds true
	\begin{align*}
	\|\bm{\Delta}_3^+\|_2 &= m\frac{|\mathcal{I^+}|-|\mathcal{I_{+\epsilon}^+}|}{|\mathcal{I_{+\epsilon}^+}||\mathcal{I^+}|}\Big\|\frac{1}{m} \sum_{i\in \mathcal{I}^+}\bm{a}_i\bm{a}_i^\ccalT \Big\|_2 \\&\le  m\frac{|\mathcal{I^+}|-|\mathcal{I_{+\epsilon}^+}|}{|\mathcal{I_{+\epsilon}^+}||\mathcal{I^+}|}\Big\|\frac{1}{m} \sum_{i=1}^m\bm{a}_i\bm{a}_i^\ccalT \Big\|_2 
	\\ &\leq \frac{8\epsilon(1+\epsilon)\kappa}{p_0(\epsilon)^2}.
	\end{align*}
	To sum, the following is true
	\begin{align}
		\label{eq:delta2+}
\|\bm{\Delta}^+\|_2 &= \|\bm{\Delta}_1^+\|_2 + \|\bm{\Delta}_2^+\|_2+ \|\bm{\Delta}_3^+\|_2 \\
&\le \epsilon + \frac{2}{p_0(\epsilon)}\cdot 4q(\kappa\epsilon)^{1/p} + \frac{8\epsilon(1+\epsilon)\kappa}{p_0(\epsilon)^2}.
	\end{align} 
	
	Taking $p = 1+ \frac{1}{\log{\frac{1}{\kappa \epsilon}}}$ and $q = 1+ \log \frac{1}{\kappa \epsilon}$, one has  $\|\bm{\Delta}^+\|_2 \le \delta^+$, with $$\delta^+:=\epsilon +\frac{8e(1-\log \kappa \epsilon)\kappa\epsilon}{p_0(\epsilon) }+ \frac{8\epsilon(1+\epsilon)\kappa}{p_0(\epsilon)^2}.$$ Since $p_0(\epsilon)$ and $\kappa$ are bounded away from 0 for sufficiently small $\epsilon>0$, $\delta^+$ approaches 0 as $\epsilon$ approaches 0. Based on the established bounds on $\|\Delta^-\|_2$ in \eqref{eq:delta2-} and $\|\Delta^+\|_2$ in \eqref{eq:delta2+}, one has
	 $$\|\lambda^- \bm{\Delta}^- + \lambda^+ \bm{\Delta}^+\|_2 \leq \lambda^+\delta^+ -\lambda^-\delta^-.$$ 
	 From Lemma~\ref{le:mp} in the Appendix, the next can be deduced 
	$${\rm{dist}}^2(\hat{\bm{d}}, \bm{d}) \leq 2- 2|\langle \hat{\bm{d}}, \bm{d} \rangle| \le  \left(\frac{2(\lambda^+\delta^+ -\lambda^-\delta^-)}{\lambda^+\phi(\tau^+) +\lambda^-\phi(\tau^-)}\right)^2$$
implying 
\begin{equation}
\label{eq:dister}
{\rm{dist}}(\hat{\bm{d}}, \bm{d}) \leq \frac{2(\lambda^+\delta^+ -\lambda^-\delta^-)}{\lambda^+\phi(\tau^+) +\lambda^-\phi(\tau^-)}.
\end{equation}

Combining $\mathcal{E}_1$ and the bound in \eqref{eq:dister} gives rise to
	\begin{align*}
	 {\rm{dist}}({\bm{z}^0}, \bm{x}) &\leq \hat{r}\,{\rm{dist}}(\hat{\bm{d}}, \bm{d})+|r-\hat{r}| \\&\leq \left(\sqrt{1+\epsilon} \,{\rm{dist}}(\hat{\bm{d}}, \bm{d})  +\sqrt{1+\epsilon}-1\right)\|\bm{x}\|_2.
	\end{align*}
	Letting $\delta_0 := \frac{2\sqrt{1+\epsilon}(\lambda^+\delta^+ -\lambda^-\delta^-)}{\lambda^+\phi(\tau^+) +\lambda^-\phi(\tau^-)}+ \sqrt{1+\epsilon}-1$, we have $ {\rm{dist}}({\bm{z}^0}, \bm{x}) \le \delta_0\|\bm{x}\|_2$. It is worth stressing that $\lim_{\epsilon\rightarrow 0} \delta_0  =0$, suggesting that $\delta_0$ can be brought arbitrarily close to $0$ by increasing $m$. 
	
	So far, it has been proved that $ {\rm{dist}}({\bm{z}^0}, \bm{x}) \le \delta_0\|\bm{x}\|_2$ on the events $\{\mathcal{E}_i\}_{i=1}^5$. The next step is to show the five events occur simultaneously with high probability. Recall that it has been shown in~\cite[Proposition 2]{duchi2017solving} that each of the events $\mathcal{E}_1, \mathcal{E}_2$, and $\mathcal{E}_4$ occurs with probability  at least $1-\exp(-c^-m)$ when $m > C^-n$. 
	
	To complete the proof, we first show that  $$P(\mathcal{E}_3) \ge 1 -  \exp\Big(\frac{-m p_0(\epsilon)^2}{2}\Big).$$
	To that end, rewrite $|\mathcal{I}_{+\epsilon}^+|$ as $$|\mathcal{I}_{+\epsilon}^+| = \sum_{i=1}^m \mathbb{1}_{\{\langle \bm{a}_i, \bm{d}\rangle^2 \ge (1+\epsilon)/2 \}}.$$ Since  $\{\mathbb{1}_{\{\langle \bm{a}_i, \bm{d}\rangle^2 \ge (1+\epsilon)/2 \}}\}_{i=1}^m$ are i.i.d. Bernoulli random variables with 
	$$P\big(\langle\bm{a}_i, \bm{d} \rangle^2 \ge (1+\epsilon)/2\big)\ge p_0(\epsilon),\quad \forall i \in [m]$$ the following holds
	$$P\Big(|\mathcal{I}_{+\epsilon}^+| \leq \frac{1}{2}m p_0(\epsilon) \Big) \leq \exp\Big(\frac{-m p_0(\epsilon)^2}{2}\Big)$$ by Hoeffding's inequality \cite{chap2010vershynin}. Therefore, $P(\mathcal{E}_3) \ge 1 -  \exp\Big(\frac{-m p_0(\epsilon)^2}{2}\Big)$.  Similar to \cite[Lemma A.6]{duchi2017solving}, it can be shown that for $$m\ge\log^2p_0(\epsilon)n/c^+\epsilon^2p_0(\epsilon)$$ with some absolute constant $c^+>0$, 
	$$P(\mathcal{E}_5|\mathcal{E}_3) \ge 1 -  \exp\left(\frac{-c^+ \epsilon^2m p_0(\epsilon)}{\log^2p_0(\epsilon)}\right).$$ Thus, setting $c_0 = \min\big\{\frac{c^+\epsilon^2 p_0(\epsilon)}{\log^2p_0(\epsilon)}, \,\frac{p^0(\epsilon)^2}{2}, \,c^-\big\}$ and $$C_0 = \max\{\log^2p_0(\epsilon)/c^+\epsilon^2p_0(\epsilon), \,C^-\} $$ confirms the assertion of   Theorem \ref{th:int}.

\subsection{Proof of Theorem 2}\label{sec:thm2}
	Some notation  used only for this section is introduced first. 
	For all $t\ge 0$, let 
	$$\bm{v}^{t+1}:=\bm{z}^t-\frac{\mu}{m}\sum_{i=1}^m w_i^t \Big(\bm{a}_i^\ccalT\bm{z}^t-\psi_i\frac{\bm{a}_i^\ccalT\bm{z}^t}{|\bm{a}_i^\ccalT\bm{z}^t|}\Big)\bm{a}_i$$ represent the estimate prior to effecting the hard thresholding operation in~\eqref{eq:iteration}. 
	The support of $\bm{x}$ and $\bm{z}^{t}$ is denoted as $\mathcal{S}$ and $\hat{\mathcal{S}}^{t}$, respectively.
	Hence, the support for the reconstruction error $\bm{h}^t:=\bm{x}-\bm{z}^{t}$ defined as $\Omega^t$ is given by $\mathcal{S}\cup \hat{\mathcal{S}}^{t}$. Additionally, let $\Omega^{t}\setminus \Omega^{t+1}$ be the difference between sets $\Omega^{t}$ and $\Omega^{t+1}$. 
Evidently, it holds that $|\mathcal{S}|=|\hat{\mathcal{S}}^{t}|= s$ for $s:=kB$, which implies $|\Omega^t|\le 2s$,  $|\Omega^{t}\setminus \Omega^{t+1}|\le 2s$, and $|\Omega^{t}\cup \Omega^{t+1}|\le 3s$, $\forall t\ge 0$. Vectors with sets as subscript, e.g., $\bm{v}_{\Omega^t}$, are formed by zeroing all entries of the vector except for those in the set.  
	
According to the triangle inequality of the vector $2$-norm, it holds that 
	\begin{align}
	\big\|\bm{x}-\bm{z}^{t+1}\big\|_2&=\big\|\bm{x}-\bm{v}^{t+1}_{\Omega^{t+1}}+\bm{v}^{t+1}_{\Omega^{t+1}}-\bm{z}^{t+1}\big\|_2\nonumber\\
	&\le \big\|\bm{x}-\bm{v}^{t+1}_{\Omega^{t+1}}\big\|_2+\big\|\bm{z}^{t+1}-\bm{v}^{t+1}_{\Omega^{t+1}}\big\|_2\nonumber\\
	&\le
	2\big\|\bm{x}_{\Omega^{t+1}}-\bm{v}^{t+1}_{\Omega^{t+1}}\big\|_2.
	\label{eq:triangle}
	\end{align}
	The last inequality in~\eqref{eq:triangle} comes from $\big\|\bm{z}^{t+1}-\bm{v}^{t+1}_{\Omega^{t+1}}\big\|_2 \le \big\|\bm{x}_{\Omega^{t+1}}-\bm{v}^{t+1}_{\Omega^{t+1}}\big\|_2$  since $\bm{z}^{t+1}$ achieves the minimal distance to $\bm{v}^{t+1}_{\Omega^{t+1}}$ among all vectors belonging to $\mathcal{M}_B^k$ and supported on $\Omega^{t+1}$. 
	Substituting the definitions of $\bm{h}^t$ and $\bm{v}^t$ into \eqref{eq:triangle}, one arrives at 
	\begin{align}
	\frac{1}{2}\|\bm{h}^{t+1}&\|_2
	\le \Big\|\bm{h}_{\Omega^{t+1}}^t-\frac{\mu}{m}\sum_{i=1}^m w_i^t\bm{a}_i^\ccalT\bm{h}^t\bm{a}_{i,\Omega^{t+1}}\nonumber\\
	 -&\frac{\mu}{m}\sum_{i=1}^m w_i^t\Big(\frac{\bm{a}_i^\ccalT\bm{z}^{t}}{|\bm{a}_i^\ccalT\bm{z}^{t}|}-\frac{\bm{a}_i^\ccalT\bm{x}}{|\bm{a}_i^\ccalT\bm{x}|}\Big)|\bm{a}_i^\ccalT\bm{x}|\bm{a}_{i,\Omega^{t+1}}\Big\|_2\nonumber\\
	\le& \Big\|\bm{h}_{\Omega^{t+1}}^t-\frac{\mu}{m}\sum_{i=1}^m w_i^t\bm{a}_{i,\Omega^{t+1}}\bm{a}_{i,\Omega^{t+1}}^\ccalT\bm{h}_{\Omega^{t+1}}^t\Big\|_2\nonumber\\
	 +&\Big\|\frac{\mu}{m}\sum_{i=1}^m w_i^t\bm{a}_{i,\Omega^{t+1}}\bm{a}_{i,\Omega^{t}\setminus\Omega^{t+1}}^\ccalT\bm{h}_{\Omega^{t}\setminus\Omega^{t+1}}^t\Big\|_2\nonumber\\
	 +&\Big\|\frac{\mu}{m}\sum_{i=1}^m w_i^t\!\Big(\frac{\bm{a}_i^\ccalT\bm{z}^{t}}{|\bm{a}_i^\ccalT\bm{z}^{t}|}\!-\!\frac{\bm{a}_i^\ccalT\bm{x}}{|\bm{a}_i^\ccalT\bm{x}|}\Big)|\bm{a}_i^\ccalT\bm{x}|\bm{a}_{i,\Omega^{t+1}}\Big\|_2\label{eq:expand}.
	\end{align}
Hence, bounding $\|\bm{h}^{t+1}\|_2$ suffices to bound the three terms on the right hand side of \eqref{eq:expand}. 
	
	Regarding the first term, the following holds
		\begin{align}
		&	\Big\|\bm{h}_{\Omega^{t+1}}^t\!-\!\frac{\mu}{m}\sum_{i=1}^m w_i^t\bm{a}_{i,\Omega^{t+1}}\bm{a}_{i,\Omega^{t+1}}^\ccalT\bm{h}_{\Omega^{t+1}}^t\Big\|_2\nonumber\\
		&\le\Big\|\bm{I}-\frac{\mu}{m}\sum_{i=1}^m w_i^t\bm{a}_{i,\Omega^{t+1}}\bm{a}_{i,\Omega^{t+1}}^\ccalT\Big\|_2\big\|\bm{h}_{\Omega^{t+1}}^t\big\|_2\nonumber\\
		&\le \max\!\big\{1-\mu\underline{\lambda},\,\mu\widebar{\lambda}-1\big\}\big\|\bm{h}_{\Omega^{t+1}}^t\big\|_2\label{eq:1stterm}
		\end{align}
		in which $\widebar{\lambda} $ and $\underline{\lambda}>0$ denote the largest and smallest eigenvalue of $(1/m)\sum_{i=1}^m w_i^t\bm{a}_{i,\Omega^{t+1}}\bm{a}_{i,\Omega^{t+1}}^\ccalT$, respectively. Since 
		$\tau_w \le w_i^t \leq 1$ and $\bm{a}_{i,\Omega^{t+1}}\bm{a}_{i,\Omega^{t+1}}^\ccalT,~\forall i\in[m]$ are positive semidefinite,  the next is true
		\begin{align} \label{eq:lambdalim}
		 \tau_w \sum_{i=1}^m \bm{a}_{i,\Omega^{t+1}}\bm{a}_{i,\Omega^{t+1}}^\ccalT&\le \sum_{i=1}^m \nonumber w_i^t\bm{a}_{i,\Omega^{t+1}}\bm{a}_{i,\Omega^{t+1}}^\ccalT \\&\le \sum_{i=1}^m \bm{a}_{i,\Omega^{t+1}}\bm{a}_{i,\Omega^{t+1}}^\ccalT.
		\end{align}
		To estimate the eigenvalues of $(1/m)\sum_{i=1}^m \bm{a}_{i,\Omega^{t+1}}\bm{a}_{i,\Omega^{t+1}}^\ccalT$, we resort to the restricted isometry property  of Gaussian matrices $\bm{A}\in\mathbb{R}^{m\times n}$ whose entries are i.i.d. standard Gaussian variables~\cite{tit2005candes}. Specifically, if $\mathcal{K}\subsetneqq\{1,\,\ldots,\,n\}$ with $|\mathcal{K}|\le 3s$,  then for constant $\delta_{3s}\le \epsilon$, the following holds for all $\bm{u}\in\mathbb{R}^m$ 
		$$\sqrt{(1-\delta_{3s})m}\|\bm{u}\|_2 \le \|\bm{A}_{\mathcal{K}}^\mathcal{T}\bm{u}\|_2\le \sqrt{(1+\delta_{3s})m}\|\bm{u}\|_2$$ 
		with probability at least $1-{\rm e}^{-c_1'm}$, provided that $m\ge C_1'\epsilon^{-2}(3s)\log(n/(3s))$ for numerical constants $c_1',\,C_1'>0$~\cite[Proposition 3.1]{acha2009nt}.  Hence,    
     \begin{align}
      \lambda_{1}\Big(\frac{1}{m}\sum_{i=1}^m\bm{a}_{i,\Omega^{t+1}}\bm{a}_{i,\Omega^{t+1}}^\ccalT\Big) \le 1+\delta_{3s}\label{eq:lambdamax}\\
     \lambda_{n}\Big(\frac{1}{m}\sum_{i=1}^m\bm{a}_{i,\Omega^{t+1}}\bm{a}_{i,\Omega^{t+1}}^\ccalT\Big)\ge 1-\delta_{3s} \label{eq:lambdamin}
     \end{align}
     due to $|\Omega^{t+1}| \le 2s$. Substituting \eqref{eq:lambdamax} and \eqref{eq:lambdamin} into \eqref{eq:lambdalim} yields
     $$\widebar{\lambda}\le 1+\delta_{3s}, \quad \text{and} \quad  \underline{\lambda}\ge\tau_w(1-\delta_{3s})$$
which together with  \eqref{eq:1stterm} suggests that
		\begin{align} \label{eq:1st}
		&	\Big\|\bm{h}_{\Omega^{t+1}}^t\!-\frac{\mu}{m}\sum_{i=1}^mw_i^t\bm{a}_{i,\Omega^{t+1}}\bm{a}_{i,\Omega^{t+1}}^\ccalT\bm{h}_{\Omega^{t+1}}^t
		\Big\|_2\nonumber\\
		&	\le \max\!\big\{1-\mu\tau_w(1-\delta_{3s}),\,\mu(1+\delta_{3s})-1\big\}\big\|\bm{h}_{\Omega^{t+1}}^t\big\|_2.
		\end{align}

	Consider now the second term in \eqref{eq:expand}.  For convenience, define
    \begin{align*}
    \bm{A}_{\Omega^{t+1}}^\ccalT &:=[\bm{a}_{1,\Omega^{t+1}}~\cdots~\bm{a}_{m,\Omega^{t+1}}] \\
    \bm{A}_{\Omega^{t+1}}^\ccalT &:=[\bm{a}_{1,\Omega^{t+1}}~\cdots~\bm{a}_{m,\Omega^{t+1}}] \\
    \bm{A}_{\Omega^{t}\cup \Omega^{t+1}}^\ccalT&:=[\bm{a}_{1,\Omega^{t}\cup \Omega^{t+1}}~\cdots~\bm{a}_{m,\Omega^{t}\cup \Omega^{t+1}}]
    \end{align*}
     and let  $\bm{W}$ be a diagonal matrix with its $i$-th diagonal entry being $w_i^t$. Then, one has
	\begin{align}
	\label{eq:2ndterm}
	&\Big\|\frac{\mu}{m}\sum_{i=1}^m w_i^t\bm{a}_{i,\Omega^{t+1}}\bm{a}_{i,\Omega^{t}\setminus\Omega^{t+1}}^\ccalT\bm{h}_{\Omega^{t}\setminus\Omega^{t+1}}^t\Big\|_2
	\nonumber\\
	&\le \Big\|\frac{\mu}{m}\bm{A}_{\Omega^{t+1}}^\ccalT \bm{W}\bm{A}_{\Omega^{t}\setminus\Omega^{t+1}}\Big\|_2
	\big\|\bm{h}_{\Omega^{t}\setminus\Omega^{t+1}}^t\big\|_2\nonumber\\
	&\le  \Big\|\frac{\mu}{m}\bm{A}_{\Omega^{t}\cup \Omega^{t+1}}^\ccalT \bm{W}\bm{A}_{\Omega^{t}\cup \Omega^{t+1}} - \mu\frac{\tau_w+1}{2}\bm{I}\Big\|_2
	\big\|\bm{h}_{\Omega^{t}\setminus\Omega^{t+1}}^t\big\|_2 \nonumber\\
	&\le\mu \frac{1-\tau_w+2\delta_{3s}}{2}\big\|\bm{h}_{\Omega^{t}\setminus\Omega^{t+1}}^t\big\|_2
	\end{align}
    where the first inequality is due to the definition of the matrix norm, the second inequality comes from the fact that $\bm{A}_{ \Omega^{t+1}}^\ccalT \bm{W}\bm{A}_{\Omega^{t}\setminus\Omega^{t+1}}$ is a submatrix of $\bm{A}_{\Omega^{t}\cup \Omega^{t+1}}^\ccalT \bm{W}\bm{A}_{\Omega^{t}\cup \Omega^{t+1}}$,  and the last inequality stems from $\tau_w <1$.

	Finally, for the last term in~\eqref{eq:expand}, let
	$\bm{v}^t:=[v_1^t~\cdots~v_m^t]^\ccalT$ with $v_i^t:=w_i^t(\frac{\bm{a}_i^\ccalT\bm{z}^{t}}{|\bm{a}_i^\ccalT\bm{z}^{t}|}-\frac{\bm{a}_i^\ccalT\bm{x}}{|\bm{a}_i^\ccalT\bm{x}|})|\bm{a}_i^\ccalT\bm{x}|,~\forall i\in [m]$. By the definition of the induced matrix $2$-norm, it holds that
	\begin{align}\label{eq:3rdterm}
	&\quad\Big\|\frac{1}{m}\sum_{i=1}^m w_i^t\Big(\frac{\bm{a}_i^\ccalT\bm{z}^{t}}{|\bm{a}_i^\ccalT\bm{z}^{t}|}-\frac{\bm{a}_i^\ccalT\bm{x}}{|\bm{a}_i^\ccalT\bm{x}|}\Big)|\bm{a}_i^\ccalT\bm{x}|\bm{a}_{i,\Omega^{t+1}}\Big\|_2
	\nonumber\\
	&=
	\frac{1}{m}\big\|\bm{A}_{\Omega^{t+1}}^\ccalT\bm{v}^t\big\|_2\le \Big\|\frac{1}{\sqrt{m}}\bm{A}_{\Omega^{t+1}}^\ccalT\Big\|_2\Big\|\frac{1}{\sqrt{m}}\bm{v}^t\Big\|_2 \nonumber \\
    & \le 	(1+\delta_{3s})\Big\|\frac{1}{\sqrt{m}}\bm{v}^t\Big\|_2.
	\end{align}
	Regarding the term $\|\frac{1}{\sqrt{m}}\bm{v}^t\|_2$, the following holds
		\begin{align}
		\label{eq:boundv}
		\frac{1}{m}\left\|\bm{v}^t\right\|_2^2 
		&=\frac{1}{m}\sum_{i=1}^m w_i^t\Big(\frac{\bm{a}_i^\ccalT\bm{z}^{t}}{|\bm{a}_i^\ccalT\bm{z}^{t}|}-\frac{\bm{a}_i^\ccalT\bm{x}}{|\bm{a}_i^\ccalT\bm{x}|}\Big)^2|\bm{a}_i^\ccalT\bm{x}|^2\nonumber\\
		&\le2\cdot \frac{1}{m}\sum_{i=1}^m|{\rm sgn}(\bm{a}_i^\ccalT\bm{z}) -{\rm sgn}(\bm{a}_i^\ccalT\bm{x})||\bm{a}_i^\ccalT\bm{x}||\bm{a}_i^\ccalT\bm{h}|\nonumber\\
		&\le  4\frac{\sqrt{1+\delta_{2s}}}{1-\rho_0}\Big(\delta_{2s}+\sqrt{\frac{21}{20}}\rho_0\Big)	\|\bm{h}\|_2^2	
		\end{align}
		where the first inequality originates from that $|\bm{a}_i^\ccalT\bm{x}|\le |\bm{a}_i^\ccalT\bm{h}^t|$ when ${\rm sgn}(\bm{a}_i^\ccalT\bm{z}) \neq {\rm sgn}(\bm{a}_i^\ccalT\bm{x})$ and $0<w_i^t<1$; the second inequality is obtained by appealing to Lemma  \ref{le:smallprob} in the Appendix, which is  adapted from \cite[Lemma 7.17]{pwf}. Taking the result in \eqref{eq:boundv} into \eqref{eq:3rdterm} gives rise to
	\begin{align}
	\label{eq:3rdtermfinal}
	&\quad	\Big\|\frac{1}{m}\sum_{i=1}^m\Big(\frac{\bm{a}_i^\ccalT\bm{z}^{t}}{|\bm{a}_i^\ccalT\bm{z}^{t}|}-\frac{\bm{a}_i^\ccalT\bm{x}}{|\bm{a}_i^\ccalT\bm{x}|}\Big)|\bm{a}_i^\ccalT\bm{x}|\bm{a}_{i,\Omega^{t+1}}\Big\|_2\nonumber\\
	&\le (1+\delta_{3s})
	\gamma\big\|\bm{h}^t\big\|_2
	\end{align}
	where $\gamma:=2\sqrt{\frac{\sqrt{1+\delta_{2s}}}{1-\rho_0}\Big(\delta_{2s}+\sqrt{\frac{21}{20}}\rho_0\Big)}$. 
	
	Plugging the bounds in~\eqref{eq:1st}, \eqref{eq:2ndterm}, and \eqref{eq:3rdtermfinal} into \eqref{eq:expand} confirms that
	\begin{align}
	\big\|\bm{h}^{t+1}\big\|_2 
	&\!\le\! 2\max \!\big\{\!1-\mu\tau_w(1-\delta_{3s}),\,\mu(1+\delta_{3s})-1\!\big\}\!\big\|\bm{h}_{\Omega^{t+1}}^t\big\|_2
	\nonumber\\
	&+\!\mu {(1\!-\!\tau_w+2\delta_{3s})}\big\|\bm{h}_{\Omega^{t}\setminus\Omega^{t+1}}^t\big\|_2\!+\!2\mu(1+\delta_{3s})
	\gamma\big\|\bm{h}^t\big\|_2\nonumber\\
	&\le \!2\sqrt{2}\max\!\big\{\!\max\!\big\{1\!-\!\mu\tau_w(1-\delta_{3s}),\,\mu(1+\delta_{3s})\!-\!1\big\},\nonumber\\
	&\quad\mu {(1-\tau_w+2\delta_{3s})/2}\big\}\|\bm{h}^t\|_2 +2\mu(1+\delta_{3s})\gamma
	\big\|\bm{h}^t\big\|_2\nonumber\\
	&\le 2\Big[ \sqrt{2}\max\!\big\{\!1-\mu\tau_w(1-\delta_{3s}),\,\mu(1+\delta_{3s})-1\!\big\}\!,\nonumber\\
	&\quad 
	\mu (1-\tau_w+2\delta_{3s})/2\big\} +\mu(1+\delta_{3s})\gamma
	\Big]\big\|\bm{h}^t\big\|_2\nonumber\\
	&:= \rho\big\|\bm{h}^t\big\|_2
	\label{eq:finalbound}
	\end{align}
	where the second inequality is due to  $$\big\|\bm{h}_{\Omega^{t+1}}^t\big\|_2+\big\|\bm{h}^t_{\Omega^{t}\setminus\Omega^{t+1}}\big\|_2\le \sqrt{2}\,\big\|\bm{h}^t\big\|_2$$ for disjoint sets $\Omega^{t+1}$ and $\Omega^{t}\setminus\Omega^{t+1}$. From~\eqref{eq:finalbound}, it is clear that for proper $\tau_w$ and sufficiently small $\rho_0$, $\delta_{2s}$, and $\delta_{3s}$, one can select a constant step size $\mu >0$ such that $\rho<1$. Theorem~\ref{sec:thm2} can be then directly deduced by combining Theorem \ref{th:int}, Lemma~\ref{le:supp}, and equation~\eqref{eq:finalbound}.  
	

\section{Concluding Remarks}
\label{sec:con}

This contribution developed a compressive reweighted amplitude flow (CRAF) algorithm for phase retrieval of (block)-sparse signal vectors. CRAF first estimates the support of the underlying signal vector, which is followed by a new spectral procedure to obtain an effective initialization. To strengthen this initial guess, CRAF proceeds with (model-based) hard thresholding iterations relying on  reweighted gradients of the amplitude-based least-squares loss function. 
 CRAF provably recovers the true signal vectors exponentially fast when a  sufficient number of measurements become available. Judicious numerical tests corroborate the merits of 
  CRAF relative to state-of-the-art solvers of the same kind.

\section*{Appendix: Supporting Lemmas}
\label{sec:app}
The following lemma which bounds the distance between the principle eigenvectors of two symmetric matrices is adapted from \cite[Corollary 1]{yu2014useful}.
\begin{lemma}\label{le:mp}~\cite[Corollary 1]{yu2014useful}~Let $\bm{Z}:= \bm{X} + \bm\Delta$ with $\bm{X}$ and $\bm{\Delta}$ being symmetric matrices, unit vectors $\bm{v}_1$ and $\bm{u}_1$ be the principal  eigenvectors of $\bm{Z}$ and $\bm{X}$, and $\theta := \cos^{-1}\langle \bm{u}_1, \bm{v}_1 \rangle$ represent the angles between $\bm{u}_1$ and $\bm{v}_1$. It then holds that
	\begin{equation}
	\sqrt{1-\langle \bm{u}_1, \bm{v}_1 \rangle^2} = |\sin \theta| \leq \frac{2\|\bm\Delta\|_2}{\lambda_1(\bm{X})- \lambda_2(\bm{X})}.
	\end{equation}
\end{lemma}

The next lemma adopted from~\cite{2017hedge} certifies that Steps 2-4 in Algorithm~\ref{alg:SPARTA} recover the true support of $\bm{x}$ with high probability. 
\begin{lemma}[Support estimate~\cite{2017hedge}]
	\label{le:supp} 
	Consider any $k$-block-sparse signal vector $\bm{x}\in\mathbb{R}^n$ with support $\mathcal{S}$ and $\bm{x}_{\min}^B~:=~\min_{b\in \mathcal{S}_B} \|\bm{x}_b\|_2^2$  on the order of $(1/k) \|\bm x\|_2^2$. Assume $\{\bm{a}_i\}_{i=1}^m$ are i.i.d standard Gaussian, that is, $\bm{a}_i\sim\mathcal{N}(\bm{0},\bm{I}_n)$. There exists an event of probability exceeding $1-6/m$ such that, Steps 2-4 in Algorithm \ref{alg:SPARTA} recover $\mathcal{S}$   if $m\ge C_0' k^2B\log(mn)$ for some positive constant $C_0'$.
\end{lemma}

The last lemma that is useful in establishing the convergence of CRAF is proved in  \cite[Lemma 7.17]{pwf}.
	\begin{lemma}~\cite[Lemma 7.17]{pwf}
		\label{le:smallprob}
		Consider $m$ noise-free measurements $\{\psi_i=|\bm{a}_i^\ccalT\bm{x}|\}_{i=1}^m$ in which  $\bm{x}\in\mathbb{R}^n$ is $s$-sparse with support $\mathcal{S}$ ,  and $\{\bm{a}_i\sim\mathcal{N}(\bm{0},\bm{I}_n)\}_{i=1}^m$ are i.i.d. sensing vectors.
	    Let  $\bm{z}\in\mathbb{R}^n$ be an $s$-sparse vector satisfying $\|\bm{z}-\bm{x}\|_2\le \rho_0\|\bm{x}\|_2 $.
		If $\bm{h}\in\mathbb{R}^n$ is $(2s)$-sparse and $m>C_3 (2s)\log(n/(2s))$ for some numerical constants $C_3$, then the following holds for $\delta_{2s}>0$
		\begin{align}
		& \frac{1}{m}\sum_{i=1}^m\left|{\rm sgn}(\bm{a}_i^\ccalT\bm{z}) -{\rm sgn}(\bm{a}_i^\ccalT\bm{x})\right||\bm{a}_i^\ccalT\bm{x}||\bm{a}_i^\ccalT\bm{h}|\nonumber\\
		&\le 2\frac{\sqrt{1+\delta_{2s}}}{1-\rho_0}\Big(\delta_{2s}+\sqrt{\frac{21}{20}}\rho_0\Big)	\|\bm{h}\|_2^2	\label{eq:smallprob}
		\end{align}
	 	 with probability exceeding $1-3{\rm e}^{-c_3 m}$ for a fixed numerical constant $c_3>0$. 
	\end{lemma}

\bibliographystyle{IEEEtran}
\bibliography{apower}

\end{document}